\documentclass[11pt,english,reqno]{amsart}

\usepackage{etex}
\usepackage{amsmath,amssymb,amsthm,bbm,mathtools,comment}
\usepackage{amsfonts}
\usepackage[shortlabels]{enumitem}
\usepackage[pdftex,colorlinks,backref=page,citecolor=blue,pdfencoding=auto, psdextra]{hyperref}
\usepackage{bookmark}

\usepackage[mathscr]{euscript}
\usepackage[usenames,dvipsnames]{color}
\usepackage{tikz,babel,adjustbox}
\usepackage[numbers]{natbib}
\usepackage{graphicx}
\usepackage{caption}
\usepackage{subcaption}
\usepackage{verbatim}
\usepackage{array}
\usepackage[frame,cmtip,arrow,matrix,line,graph,curve]{xy}
\usepackage{etoolbox}
\usepackage{pifont}
\usepackage[final]{microtype}
\usepackage{cmtiup}
\usepackage{cleveref}

\allowdisplaybreaks

\hypersetup{pdfpagemode=UseNone,pdfstartview={XYZ null null 1.00}}
\usetikzlibrary{shapes.misc,calc,intersections,patterns,decorations.pathreplacing}
\usetikzlibrary{arrows,arrows.meta,shapes,positioning,decorations.markings}
\tikzstyle arrowstyle=[scale=1]

\allowdisplaybreaks

\makeatletter
\def\@settitle{\begin{center}%
		\bfseries\Large
		\@title
	\end{center}%
}
\patchcmd{\@setauthors}{\MakeUppercase}{\normalsize}{}{}
\makeatother

\theoremstyle{plain}
\newtheorem{theorem}{Theorem}[section]
\newtheorem{lemma}[theorem]{Lemma}
\newtheorem{claim}[theorem]{Claim}
\newtheorem{proposition}[theorem]{Proposition}

\newtheorem{definition}[theorem]{Definition}
\newtheorem{hypothesis}[theorem]{Hypothesis}

\newcommand{\F}{\mathbb{F}}
\newcommand{\NN}{\mathbb{N}}

\newcommand{\pr}{\mathbb{P}}

\newcommand{\E}[0]{\mathbb{E}}

\newcommand{\beq}[1]{\begin{equation}\label{#1}}
\newcommand{\enq}[0]{\end{equation}}

\newcommand{\D}[0]{{\mathcal D}}

\renewcommand{\U}[0]{{\mathcal U}}

\newcommand{\W}[0]{{\mathcal W}}

\renewcommand{\C}[2]{{\binom{#1}{#2}}}

\usepackage{xspace}
\renewcommand{\C}{\mathcal{C}}
\renewcommand{\U}{\mathcal{U}}
\newcommand{\maxcover}{\mbox{\sf MaxCover}\xspace}
\newcommand{\clique}{$k$-$\mathsf{Clique}$\xspace}
\newcommand{\mincover}{\texorpdfstring{$k$-$\mathsf{MinCoverage}$}{k-MinCoverage}\xspace}
\newcommand{\intersection}{\texorpdfstring{$k$\nobreakdash-$\mathsf{SetIntersection}$}{k-SetIntersection}\xspace}
\newcommand{\setcover}{$\mathsf{SetCover}$\xspace}

\newcommand{\CNF}{$\mathsf{CNF}$\xspace}
\newcommand{\SAT}{$\mathsf{SAT}$\xspace}
\newcommand{\Wone}{$\mathsf{W[1]}\neq\mathsf{FPT}$\xspace}
\newcommand\ETH{\texorpdfstring{$\mathsf{ETH}$}{ETH}\xspace}
\newcommand\SETH{\texorpdfstring{$\mathsf{SETH}$}{SETH}\xspace}

\newcommand\polylog{\text{polylog}}
\newcommand\poly{\text{poly}}

\newcommand{\NP}{$\mathsf{NP}$\xspace}

\newcommand{\TGC}{$\mathsf{TGC}$\xspace}
\newcommand{\PGC}{$\mathsf{PGC}$\xspace}
\renewcommand{\W}{\textsf{W}}
\newcommand{\FPT}{$\mathsf{FPT}$\xspace}

\DeclarePairedDelimiter\abs{\lvert}{\rvert}                      

\DeclareMathAlphabet\mathbfcal{OMS}{cmsy}{b}{n}

\renewcommand*{\k}{\mathbb{F}_q}
\newcommand*{\kbar}{\overline{\mathbb{F}}_q}
\newcommand*{\A}{\mathbb{A}}

\newcommand*{\vy}{\mathbf{y}}
\newcommand*{\vf}{\mathbf{f}}
\newcommand*{\vr}{\vec{r}}
\newcommand{\oFF}{\overline{\F}}

\renewcommand{\P}{$\mathsf{P}$\xspace}

\title{Applications of Random Algebraic Constructions to\\
Hardness of Approximation}

\author{Boris Bukh}
\address{Department of Mathematical Sciences,	Carnegie Mellon University, Pittsburgh, PA 15213, USA}
\email{bbukh@math.cmu.edu}

\author{Karthik C.\ S.}
\address{Department of Computer Science, Rutgers University, Piscataway, NJ 08854, USA}
\email{karthik.cs@rutgers.edu}

\author{Bhargav Narayanan}
\address{Department of Mathematics, Rutgers University, Piscataway, NJ 08854, USA}
\email{narayanan@math.rutgers.edu}

\begin{document}

\maketitle\thispagestyle{empty}

\begin{abstract}
    In this paper, we show how one may (efficiently) construct two types of extremal combinatorial objects whose existence was previously conjectural.
    \begin{itemize}
        \item Panchromatic Graphs: For fixed $k\in\mathbb N$, a $k$-panchromatic graph is, roughly speaking, a balanced bipartite graph with one partition class equipartitioned into $k$ colour classes in which the common neighbourhoods of panchromatic $k$-sets of vertices are much larger than those of $k$-sets that repeat a colour. The question of their existence was raised by Karthik and Manurangsi~[\emph{Combinatorica}~2020]. 
        \item Threshold Graphs: For fixed $k\in\mathbb N$, a $k$-threshold graph is, roughly speaking, a balanced bipartite graph in which the common neighbourhoods of $k$-sets of vertices on one side are much larger than those of $(k+1)$-sets. The question of their existence was raised by Lin~[\emph{JACM}~2018].
    \end{itemize}

Concretely, we provide probability distributions over graphs from which we can efficiently sample these objects in near linear time. These probability distributions are defined via varieties cut out by (carefully chosen) random polynomials, and the analysis of these constructions relies on machinery from algebraic geometry (such as the Lang--Weil estimate, for example). The technical tools developed to accomplish this might be of independent interest.

As applications of our constructions, we show the following conditional time lower bounds on the parameterized set intersection problem where, given a collection of $n$ sets over universe $[n]$ and a parameter $k$, the goal is to find $k$ sets with the largest intersection.
\begin{itemize}
    \item Assuming \ETH, for any computable function $F\colon\mathbb{N}\to\mathbb{N}$, no $n^{o(k)}$-time algorithm can approximate the parameterized set intersection problem up to factor $F(k)$. This improves considerably on the previously best-known result under \ETH due to Lin~[\emph{JACM}~2018], who ruled out any $n^{o(\sqrt{k})}$ time approximation algorithm for this problem.
    \item Assuming \SETH, for every $\varepsilon>0$ and any computable function $F\colon\mathbb{N}\to\mathbb{N}$, no $n^{k-\varepsilon}$-time algorithm can approximate the parameterized set intersection problem up to factor $F(k)$. No result of comparable strength was previously known under \SETH, even for solving this problem exactly. 
\end{itemize}
Both these time lower bounds are obtained by composing panchromatic graphs with instances of the coloured variant of the parameterized set intersection problem (for which tight lower bounds were previously known). 

\end{abstract}

\clearpage 

\section{Introduction}
Over the last five decades, a symbiotic relationship has developed between the areas of  extremal combinatorics and complexity theory (broadly construed); see the wonderful book of Jukna~\cite{J11} or one of the surveys of Alon~\cite{A03,A08,A16,A20} for various applications of extremal combinatorial objects to proving lower bounds in theoretical computer science. In particular, this synergistic exchange with extremal combinatorics can be explicitly seen in subareas such as  circuit/formula lower bounds~\cite{BGKRSW96,JuknaS13}, communication complexity~\cite{CFL83,KN97,GKR16}, error correcting codes~\cite{S96,ABV01,GUV09}, and derandomization~\cite{AGHP92,NSS95,C16,CZ19}. 

In this paper, our first goal is to prove the existence of certain extremal bipartite graphs, namely threshold graphs and panchromatic graphs. The question of their existence was motivated by applications in hardness of approximation, and our second goal is to prove, using these graphs, conditional time lower bounds on the parameterized set intersection problem. Our constructions will rely crucially on random polynomials, and our third goal here is to prove various results, likely of independent interest, about the common zeroes of random polynomials over finite fields. Before we can state our results, it will help to have some background, to which we now turn.


Over the last few years, a new area in theoretical computer science, namely \emph{hardness of approximation in \P},  has benefited significantly from some of the deep results in extremal combinatorics. Hardness of approximation in \P, roughly speaking, maybe treated as the union of two subareas, namely, hardness of approximation in \emph{parameterized complexity}\footnote{We only consider the computational problems contained in the complexity class \textsf{XP} while making this statement and also think of the parameter as fixed/constant.}  and hardness of approximation  in \emph{fine-grained complexity}. 

In parameterized complexity, one studies the computational complexity of problems  with respect to multiple parameters of the input or output. For example, in the \intersection problem, we are given a collection of $n$ sets over the universe $[n]$ and a parameter $k$ as input, and the goal is to find $k$ sets in the collection which maximize the intersection size. A problem (with inputs of size $n$, along with a parameter $k$) is said to be fixed parameter tractable if it can be solved by an algorithm  running in time $T(k)\cdot \poly(n)$ for some computable function $T$. In many interesting cases, including for the \intersection\ problem, assuming the \W[1]$\neq$\FPT\ hypothesis, it is possible to show that no such algorithm exists i.e., that the problem is not  fixed parameter tractable. In light of this, one could then ask for approximation algorithms. In the case of \intersection, the task would then be to design an approximation algorithm running in time  $T(k)\cdot \poly(n)$ that can find $k$ sets in the collection whose intersection size is at least $1/F(k)$ of the intersection size of the optimal solution for some pair of computable functions $T$ and $F$. Inapproximability results in parameterized complexity aim to typically rule out such algorithms (under the \W[1]$\neq$\FPT\ hypothesis) for various classes of functions $F$; a notion particularly relevant to this paper is that of \emph{total \FPT inapproximability}, in which we rule out $F(k)$-approximation algorithms running in $T(k)\cdot \poly(n)$ time for all computable functions $T$ and $F$. We refer the reader to the textbooks~\cite{DowneyF13,CyganFKLMPPS15} for an excellent introduction to the area.

In fine-grained complexity, one aims to refine the Cobham--Edmonds thesis~\cite{edmonds_1965,cook_1970} by trying to understand the  exact time required to solve problems in \P, by basing their conditional time lower bounds  on several plausible (and popular) conjectures such as \SETH and \ETH (see Section~\ref{sec:prelim} for definitions). For example, \intersection can be na\"ively solved by exhaustive search, i.e., by computing the intersection sizes of all $k$-tuples of sets from the given collection of $n$ sets; can we do any better? For instance, is there an algorithm running in time $n^{o(k)}$ that can solve \intersection? Or even less ambitiously, is there an algorithm running in time $n^{k-0.1}$ that can solve \intersection? The theory of fine-grained complexity aims to rule out such algorithms, and inapproximability results in this area aim to prove the same conditional time lower bounds, but now against approximation algorithms. We should emphasise that the area of fine-grained complexity is not simply about proving tighter running time lower bounds for problems considered in parameterized complexity; fine-grained complexity has been successful in explaining the complexity of problems such as closest pair in a point-set~\cite{AlmanW15,R18,DKL19,KM20}, edit distance between strings~\cite{BI18,AbboudHWW16}, and  all pairs shortest paths~\cite{WilliamsW18}, amongst others, all examples of problems usually considered without any fixed parameters. We direct the interested reader to two recent surveys~\cite{RW19, FKLM20} on hardness of approximation in \P for a detailed overview of the area.

A major difficulty addressed by results in hardness of approximation in \P is that of generating a gap\footnote{There are many results in parameterized and fine-grained inapproximability under gap assumptions such as the Gap Exponential Time Hypothesis~\cite{MR16,D16} and Parameterized Inapproximability Hypothesis~\cite{LRSZ20}. In these results the gap is inherent in the assumption, and the challenge is to construct gap-preserving reductions. These results are not the focus of this paper and we shall not elaborate further on them, and the interested reader may see the recent survey~\cite{FKLM20} for more details.}, i.e.,  one must start with a hard problem with no gap (for which the time lower bound is only against exact algorithms) and reduce it to a problem of interest while generating a non-trivial gap in the process. One of the main approaches to generate the aforementioned gap, and the motivation behind our construction of threshold graphs, 
is the \emph{Threshold Graph Composition} (\TGC) framework introduced in the breakthrough work of Lin~\cite{Lin18} to show the total \FPT inapproximability of the \intersection problem. This technique was later used to prove the first non-trivial inapproximability result for the $k$-\setcover problem~\cite{CL19}, and in the proof of the current state-of-the-art inapproximability result for the same~\cite{Lin19}. Moreover, the result on the \intersection problem in~\cite{Lin18} was used by Bhattacharyya et al.~\cite{BBEGKLMM21} as the starting point to prove inapproximability results for problems in coding theory such as  the $k$-\textsf{Minimum Distance} problem (a.k.a.\ $k$-\textsf{Even Set} problem) and  the $k$-\textsf{Nearest Codeword}  problem, and for lattice problems such as the $k$-\textsf{Shortest Vector} problem and the  $k$-\textsf{Nearest Vector} problem. 

At a very high level, in \TGC, we compose an instance of the input problem that has no gap, with an extremal combinatorial object called a \emph{threshold graph} (see Section~\ref{sec:intrographs} for definitions), to produce a gap instance of the desired problem. The two main challenges in using this framework are to construct the requisite threshold graph, and to find the right way to compose the input and the threshold graph. Our construction of threshold graphs will address the first of these challenges.

Another key issue that often arises in proving conditional time lower bounds for problems in \P\  is the following. When trying to prove time lower bounds for a particular problem, it is often  natural (and sometimes seemingly necessary) to first prove the lower bound for a coloured version of the same problem, and then reduce it to the uncoloured version of the problem. For instance, if we would like to prove lower bounds based on \SETH for a  problem $\Psi$, then it is almost always the case that we first divide the variable set of size $n$ (of the \SAT formula arising from the \SETH assumption) into $k$ equal parts and reduce the problem of deciding \SAT to a problem in \P\ where, given as input $k$ collections each containing $2^{n/k}$ partial assignments to the subset of $n/k$ variables in that part, we would like to find one partial assignment from each collection that, when stitched together, forms a full \emph{satisfying} assignment to the original \SAT instance. From this problem (in \P), if we would like to reduce to $\Psi$, it is often convenient (and sometimes imperative) to first reduce to a $k$-coloured version of $\Psi$, and then reduce this coloured version to $\Psi$ itself. This final task is sometimes easy, such as for problems like $k$-\setcover or $k$-\textsf{OrthogonalVectors}, but often non-trivial, such as for \intersection or closest pair in a point-set. It is worth reiterating here that in the other direction, reducing the uncoloured problem to its coloured version is almost always easy; typically, one can reduce the uncoloured variant to its coloured counterpart via the celebrated colour coding technique of Alon, Yuster and Zwick~\cite{alon1995colour}. 

In~\cite{DKL19,KM20}, the authors proposed the \emph{Panchromatic Graph Composition} (\PGC) framework to  address this issue, and this serves as the motivation behind our construction of panchromatic graphs (see Section~\ref{sec:intrographs} for definitions). In particular, they outlined how these panchromatic graphs, assuming that they exist, can be composed with the coloured version of a problem to reduce it to the uncoloured version of the same problem. Also, it is worth noting that the same issue arises in proving time lower bounds against approximation algorithms as well, i.e., it is often easier to prove hardness of approximation results for coloured versions of problems than for their uncoloured counterparts. With this in mind, it is desirable to have panchromatic graphs with certain additional gap properties so that we can design gap preserving reductions between problems.  Our construction of panchromatic graphs will address all of these challenges.

In summary, the role of extremal combinatorial objects in the existing literature on hardness of approximation in \P\ is twofold: threshold graphs are used in the \TGC framework to generate gaps in hard problem instances, and panchromatic graphs are used in the \PGC framework to reduce hard instances of coloured variants of various computational problems to their uncoloured (computationally easier) counterparts.

\subsection{Our Contributions}
Our contributions are primarily twofold. First, in Section~\ref{sec:intrographs}, we show how to efficiently construct threshold graphs and panchromatic graphs; even the existence of such graphs was previously conjectural. Second, in Section~\ref{sec:introintersection}, we demonstrate some applications of these graphs (with panchromatic graphs featuring more prominently) to prove \emph{tight} conditional time lower bounds under \ETH and \SETH for approximating \intersection. Finally, in Section~\ref{sec:introbig} we briefly detail how our results fit into the bigger picture of hardness of approximation in \P. 

\subsubsection{Constructions of Panchromatic and Threshold Graphs}\label{sec:intrographs}
Here, we describe our main combinatorial results that demonstrate the existence of the aforementioned extremal bipartite graphs. 

We start with panchromatic graphs. 

\begin{definition}[Panchromatic Graphs; Informal version of Definition~\ref{def:pan}]
An \emph{$(n,k,t,s)$-panchromatic graph} is a bipartite graph $G(A,B)$ where $A$ is partitioned into $k$ parts, say $A_1,\ldots,A_k$, with $|A_1| = \dots =|A_k|=|B|=n$ satisfying the following pair of conditions.
\begin{description}	
\item[Completeness]
Every $k$-set $\{a_1,\dots,a_k\}$ with $a_i \in A_i$ for $i \in [k]$ has at most $t$ common neighbours in $B$, and a positive fraction (depending only on $k$)  of such $k$-sets have exactly $t$ common neighbours in $B$.
\item[Soundness] For every set $X \subset A$ of size $k$ for which $A_i \cap X$ is empty for some $i\in [k]$, the number of common neighbours of $X$ in $B$ is at most $s$.
\end{description}
\end{definition}

In~\cite{KM20}, the authors studied panchromatic graphs\footnote{\label{km20}The term `panchromatic graph' was not introduced in~\cite{KM20}. There, the authors constructed dense balanced bipartite graphs with low \emph{contact dimension}, but that construction can be reinterpreted as construction of panchromatic graphs when $k=2$; see Section 8 in~\cite{KM20}.} when $k=2$. Using (non-trivial) density properties of Reed--Solomon codes and Algebraic-Geometric codes, they were able to show that $(n,2,t,t^{o(1)})$-panchromatic graphs exist
for $t = 2^{(\log n)^{1-o(1)}}$, and that they can be constructed efficiently. They then raised the natural question of existence for general $k$, indicating that if such graphs exist, they could then potentially be used to improve hardness and inapproximability results for \intersection. We resolve this open problem from~\cite{KM20} and prove the following result.
\begin{theorem}[Informal restatement of Theorem~\ref{thm:pan_cons}]\label{thm:intropan}
For each $k \in \NN$ and any integer $\lambda > 1$, there exist $(n, k, t, t/\lambda)$-panchromatic graphs 
for infinitely many $n \in \NN$, where $t = t(k,\lambda)>0$  depends only on $k$ and $\lambda$.
\end{theorem}

In~\cite{KM20}, the authors note that their technique to construct panchromatic graphs is limited to the case of $k=2$, and remark that one needs to construct objects with more structure than just \emph{maximum distance separable} codes in a certain sense\footnote{To quote~\cite{KM20}, \emph{``The issue in constructing this graph is that we are now concerned about agreements of more than two vectors, which does not correspond to error-correcting codes anymore and some additional tools are needed to argue for this more general case.''}}. Our construction, detailed in Section~\ref{sec:introtechgraph}, does just this, introducing new ideas that go beyond standard coding-theoretic properties. 

On a different note, it is natural to ask if the requirement in the completeness condition that a positive fraction (depending on $k$) of $k$-sets have exactly $t$-sized common neighbourhoods can be strengthened to demand the same of \emph{every} such $k$-set. It turns out that our result is in fact best-possible in the following sense: as $n \to \infty$ and for any $t = t(k)$, there do not exist $(n, k, t, t-1)$\nobreakdash-panchromatic graphs in which a $\left(1-1/t\right)$-fraction of the panchromatic $k$-sets have exactly $t$-sized common neighbourhoods; this may be shown using the K\"ov\'ari–-S\'os–-Tur\'an theorem and H\"older's inequality, but we omit the details here.

Next, we turn our attention to threshold graphs. 
\begin{definition}[Threshold Graphs; Informal version of Definition~\ref{def:thres}]
An \emph{$(n,k,t,s)$-threshold graph} is a bipartite graph $G(A,B)$ with $|A|=|B|= n$ satisfying the following pair of conditions.
\begin{description}	
\item[Completeness] For every $k$-set of vertices $X \subset A$, the number of common neighbours of $X$ in $B$ is at least $t$.
\item[Soundness] For every $(k+1)$-set of vertices $X \subset A$, the number of common neighbours of $X$ in $B$ is at most $s$.
\end{description}
\end{definition}

These graphs are closely related to constructions for Tur\'an-type problems in extremal graph theory. Indeed, if the completeness condition above is weakened to only require that a positive fraction (depending on $k$) of $k$-sets $X \subset A$ have at least $t$ common neighbors in $B$, then the celebrated norm-graphs of~\cite{KollarRS96,BGKRSW96} achieve these weakened requirements. 

Lin~\cite{Lin18} raised the question\footnote{To quote~\cite{Lin18}, \emph{``However, at the moment of writing, I do not know how to do that,
even probabilistically.''}} of the existence of threshold graphs, and noted that if threshold graphs exist, then there is a very short proof\footnote{Starting with an instance $G_0(V_0,E_0)$ of the canonical \W[1]-hard $k$-clique problem on $n$ vertices, we combine it with a $(n,k,t,s)$-threshold graph $G(V_0,B)$ to yield an instance of $\binom{k}{2}$-\textsf{SetIntersection} with $|E_0|$ sets on the universe $B$, where for every edge $e=(u,v)\in E_0$, we include the element $b\in B$ in the set associated with this edge if and only if $b$ is a common neighbor of $u$ and $v$ in $G$. It then follows that if there is a $k$-clique in $G_0$, then there are $\binom{k}{2}$ sets whose intersection size is at least $t$, and if there is no $k$-clique in $G_0$, then every $\binom{k}{2}$ sets have intersection size at most $s$.\label{short}} of the total \FPT inapproximability of \intersection. However, since the existence of threshold graphs was previously unknown, the argument showing total \FPT inapproximability of \intersection in~\cite{Lin18} is rather delicate. We resolve this open problem from~\cite{Lin18} and show that threshold graphs exist, obtaining a very short proof of the total \FPT inapproximability of \intersection as a byproduct. 

\begin{theorem}[Informal restatement of Theorem~\ref{thm:thresh_cons}]\label{thm:introthre}
For each $k \in \NN$ and 
for infinitely many $n \in \NN$, there exist $(n, k, n^{\Omega(1/k)}, k^{O(k)})$-threshold graphs. 
\end{theorem}
The parameters in this result match the parameters obtainable via norm-graphs, but crucially, our construction also achieves the stronger completeness property discussed earlier.
It is possible to improve the $k^{O(k)}$ to $2^{O(k)}$ using the arguments in \cite{bukh_exponential_turan}, but
we avoid the extra complexity of that approach.

\subsubsection{Applications to Parameterized Set Intersection Problem}\label{sec:introintersection}
Here, we describe our  conditional time lower bounds for the \intersection problem. In order to set the context for the complexity of this problem, we briefly recall its complexity in the world of \NP. 

In the world of complexity, \textsf{SetIntersection} is well-known as a notorious problem to prove any kind of hardness of approximation result for; that said, there is a general belief that it is a hard problem as no non-trivial polynomial time approximation algorithms for this problem are known. However, to this date, even ruling out a \textsf{PTAS} under the standard \P$\neq$\NP\ hypothesis remains open!\footnote{In contrast, it is fairly straightforward to show that the exact version of the problem is \NP-hard~\cite{X12}.} The best inapproximability result for this problem is based on assuming that \SAT problems of size $n$ cannot be solved by randomized algorithms in time $2^{n^\varepsilon}$, under which Xavier~\cite{X12} shows that there is no  polynomial time algorithm which can approximate \textsf{SetIntersection} up to polynomial factor. It is worth noting that to prove this inapproximability result, the author indirectly relies on the highly non-trivial and celebrated quasi-random \textsf{PCP} construction of Khot~\cite{Khot06}.

Given this context, it was truly a breakthrough when Lin~\cite{Lin18}, introducing some novel techniques, proved the total \FPT inapproximability of \intersection (under \Wone hypothesis). Of course, using our construction of threshold graphs (Theorem~\ref{thm:introthre}), we now have a very short proof of this powerful result (see \cref{short}). Lin~\cite{Lin18} further refined his inapproximability result and showed, assuming \ETH, that for sufficiently large $k\in\NN$, no randomized $n^{o(\sqrt{k})}$-time algorithm can approximate \intersection to a factor $n^{1/\Omega(\sqrt{k})}$. Clearly, this result is stronger than ruling out $F(k)$ approximation algorithms (for some function $F$), but the running time lower bound is far from tight. The following result, the first application of our constructions, shows that we can improve on Lin's result and obtain tight running time lower bounds under \ETH (albeit for weaker approximation factors). 

\begin{theorem}[Informal restatement of Theorem~\ref{thm:ETH}]\label{thm:introETH}
Let $F\colon \NN\to\NN$ be any computable function.   Assuming (randomized) \ETH, for sufficiently large $k\in\NN$, no randomized $n^{o(k)}$-time algorithm can approximate \intersection to a factor $F(k)$. 
\end{theorem}

In the world of fine-grained complexity, it is also of interest to prove, under stronger assumptions than \ETH, even tighter running time lower bounds than the $n^{o(k)}$ bound above. In particular, one would like to rule out $n^{k-0.1}$-time algorithms for \intersection under \SETH, essentially showing that the na\"ive exhaustive search algorithm for \intersection is optimal. To the best of our knowledge, it was not known earlier if one could even rule out \emph{exact} algorithms for \intersection running in $n^{k-0.1}$-time under \SETH. We remedy this situation; the following strong inapproximability result under \SETH is the second application of our constructions. 

\begin{theorem}[Informal restatement of Theorem~\ref{thm:SETH}]\label{thm:introSETH}
Let $F\colon \NN\to\NN$ be any computable function.   Assuming (randomized) \SETH, for every $\varepsilon > 0$ and integer $k > 1$, no randomized $n^{k(1-\varepsilon )}$-time algorithm can approximate \intersection to a factor $F(k)$. 
\end{theorem}

Both of these results are crucially reliant on our construction of panchromatic graphs; a broad outline is given in Section~\ref{sec:introtechseth}. It is worth noting that for the coloured variant of \intersection, one can easily show tight running time lower bounds under \ETH and \SETH against exact algorithms, and by using non-trivial gap creating techniques, these tight running time lower bounds were extended against near polynomial factor approximation algorithms for the coloured variant in~\cite{KLM19}. The situation (for the coloured variant) is similar in the world of \NP as well; see~\cite{CP11}. Finally, we remark that by using the hardness of approximation results in~\cite{KLM19} under the $k$-\textsf{SUM} hypothesis, we can use the \PGC framework to  rule out randomized $n^{k(1/2-\varepsilon )}$-time $F(k)$-factor approximation algorithms for \intersection under the $k$-\textsf{SUM} hypothesis.

\subsubsection{Bigger Picture: Reverse Colour Coding}\label{sec:introbig} 
We conclude this discussion of our results by briefly highlighting a broader implication. For many computational problems, it is often natural to define and study a coloured variant. For some problems, the coloured variant turns out to be even more natural; for example, any $k$-\textsf{CSP} (i.e., constraint satisfaction problems of arity $k$) on $k$ variables can be seen as a coloured version of the maximum edge biclique problem. Establishing computational equivalences between coloured and non-coloured variants of problems is thus a basic question worthy of exploration. As noted earlier, for some problems, there is a straightforward equivalence between the two versions. However, there are many important problems for which this equivalence is nontrivial (and potentially not true). The celebrated colour coding technique of Alon, Yuster and Zwick~\cite{alon1995colour} provides an efficient way for a problem to be reduced to its coloured variant. Our construction of panchromatic graphs (when combined with \PGC, as will be described in Section~\ref{sec:introtechseth}) now gives us a rather general method to reverse the colour coding technique.

\subsection{Our Techniques}
Our main technical contribution is the constructions of panchromatic graphs and threshold graphs which we describe in Section~\ref{sec:introtechgraph}. We also provide an overview of how these are used to prove Theorems~\ref{thm:introETH}~and~\ref{thm:introSETH} in Section~\ref{sec:introtechseth} 

\subsubsection{Constructions of Panchromatic and Threshold Graphs}\label{sec:introtechgraph}

To motivate our approach, we start by explaining, briefly, why a natural first attempt at constructing threshold graphs fails. It is natural to consider a random bipartite graph where each edge is included independently with an appropriately chosen probability $p$. Indeed, it is easy to see that such a construction can ensure that \emph{most} $k$-sets of vertices on one side have fewer common neighbours than \emph{most} $(k+1)$-sets. However, it is essentially impossible to avoid some \emph{exceptional} $k$-sets and $(k+1)$-sets at the relevant edge density $p$. Without getting into the details, the reason for this is simple: the size of the common neighbourhoods in this probability space have long, smoothly-decaying tails, and since there are many sets to consider, it is overwhelmingly likely that exceptional sets exist. For more on this issue, we refer the reader to~\cite{bukh}. 

When it comes to panchromatic graphs, while there is no immediate natural candidate construction, it seems clear that assuming one wishes to construct such objects randomly, one needs to introduce some level of correlation between different edges, while simultaneously preserving enough independence to allow us to analyse the resulting random graph, a delicate task from a purely probabilistic perspective.

It turns out that there is a natural way to circumvent all the obstacles outlined above, namely, by considering random graphs in which adjacency is determined by a randomly chosen algebraic variety. Concretely, our approach, which works over the finite field $\k$ for any prime power $q \in \NN$, is as follows.

\begin{enumerate}
    \item We construct threshold graphs as follows. We build $A$ by independently sampling $q^{k+1}$ random polynomials of degree $d$ from $\F_q[X_1, \dotsc, X_{k+1}]$ for a suitable $d = d(k)$. Then, with $B = \F_q^{k+1}$, we define a bipartite graph $G$ between $A$ and $B$ by joining $f \in A$ to $x \in B$ if $f(x) = 0$.
    \item To construct panchromatic graphs, we proceed as follows. First, we independently choose random polynomials $w_1, \dots, w_k$ of degree $D$ from $\F_q [X_1, \dotsc, X_k]$ for a suitable $D = D(k)$. Next, for $i \in [k]$, we take $A_i$ to be a set of $q^k$ random polynomials of the form $w_i + p$, where each such $p$ is an independently sampled random polynomial of degree $d$ from $\F_q [X_1,\dotsc, X_k]$ for a suitable $d = d(k)$. Finally, with $B = \F_q^k$, we define a bipartite graph $G$ between $A$ and $B$ by joining $f \in A$ to $x \in B$ if $f(x) = 0$.
\end{enumerate}

While the random algebraic graphs above are quite easy to describe, their analysis is far from simple; in particular, to prove our main results, we shall rely on Lang--Weil estimate~\citep{LW54}, which is a consequence of the Riemann hypothesis for function fields (but see~\citep{schmidt_book} for a relatively elementary proof).
Along the way, we shall prove a several results about the zero sets of random polynomials over finite fields that may be of independent interest. An illustrative example is the following probabilistic analogue of B\'ezout's theorem over finite fields.

\begin{theorem}
  For $k,d \in \NN$ and a prime power $q \in \NN$, let $Z$ be the (random) number of common roots over $\k^k$ of $k$ independently chosen $k$-variate random $\k$-polynomials of degree $d$. Then, as $q \to \infty$, we have
  \[ \pr[Z = d^k] \ge \frac{1-o(1)}{(d^k)!},\]
  as well as
  \[ \pr[Z > d^k] =O( q^{-d}).\]
\end{theorem}

To place these techniques in context, it is worth mentioning that the first traces of this random algebraic method go back some way, to work of Matou\v sek~\cite{Mat} in discrepancy theory, but it is the variant originating in~\cite{bukh} and developed further in~\citep{bc, conlon} that we shall build upon in this paper.

\subsubsection{Hardness of Approximating \intersection}\label{sec:introtechseth}
The common starting point for both Theorems~\ref{thm:introETH}~and~\ref{thm:introSETH} is the {\sf Unique} $k$-\maxcover problem defined in~\cite{KLM19}. We refrain from defining it here, but it is immediate from its definition (see Section~\ref{sec:prelim}) that it can be easily reformulated as the coloured version of \intersection (see Proposition~\ref{prop:mctointer}), hereafter panchromatic \intersection.  In panchromatic \intersection, we are given $k$ collections, each consisting of $n$ subsets of the universe $[n]$, and the goal is to choose one set from each collection such that their intersection size is maximized. From~\cite{KLM19}, it follows that assuming \SETH (respectively \ETH), there is no $n^{k-\varepsilon}$-time (respectively $n^{o(k)}$-time) algorithm  that can approximate panchromatic \intersection to an $F(k)$ factor for any computable function $F$. 

It is easier to describe the \PGC technique in terms of graphs, so we reformulate the panchromatic \intersection problem as follows: given a bipartite graph $H(X, Y)$ where $X=X_1\dot\cup\cdots \dot\cup X_k$ corresponds to the $k$ collections of sets and $Y$ corresponds to the universe (so $|X_1|=\cdots =|X_k|=|Y|=n$), the goal is to find $(x_1,\ldots ,x_k)\in X_1\times \cdots \times X_k$ which has the largest sized common neighbourhood in $Y$. We also consider a $(n,k,t,t/\lambda)$-panchromatic graph $G(X,B)$ as guaranteed by our Theorem~\ref{thm:intropan}. Now, given $G$ and $H$ as above, the \PGC technique, roughly speaking, boils down to analyzing the graph $H^*(X,Y\times B)$ where if $(x,b)\in X_i\times B$ is an edge in $G$ and $(x,y)\in X_i\times Y$ is an edge in $H$, then we have the edge $(x,(y,b))\in X_i\times Y\times B$ in $H^*$.

In the completeness case, if the maximum panchromatic common neighbourhood size in $H$ was $c$, then the same set of vertices would have a common neighbourhood of size $t\cdot c$ in $H^*$, whereas in the soundness case, if the maximum panchromatic common neighbourhood size in $H$ was $s$, then the maximum common neighbourhood size is at most $t\cdot s$ in $H^*$. From the soundness of the panchromatic graph, we know that if we pick $k$ vertices in $X$ not all from different colour classes, then their common neighbourhood is of size at most $(t/\lambda)\cdot |Y|$. The results we desire then follow by setting $\lambda$ appropriately, and importantly noting that $|Y|=O(c)$ in the hard instances given by~\cite{KLM19}; recall that the common neighbourhood problem on $H^*$ where we ignore the colour classes is the \intersection problem.

Our composition technique using panchromatic graphs strictly improves on the techniques introduced in~\cite{DKL19,KM20}. The \PGC technique described above also improves the inapproximability results of~\cite{KM20}, albeit only in the lower order terms, and also simplifies their hardness of approximation proof for the \emph{Monochromatic Maximum Inner Product} problem.  

\subsection{Organization of Paper}
In Section~\ref{sec:prelim}, we formally define the problems and hypotheses of interest in this paper. In  Section~\ref{sec:graphs},
we carefully define panchromatic and threshold graphs and state our main results about them. 
In Section~\ref{sec:tools}, we prove some important intermediate results that will be used to analyze our constructions of panchromatic and threshold graphs. 
In Section~\ref{sec:cons}, we give the constructions of panchromatic graphs and threshold graphs. In
Section~\ref{sec:inter}, we prove our fine-grained inapproximability results for \intersection.
Finally, in 
Section~\ref{sec:open}
we highlight a few important open problems and research directions.

\section{Preliminaries}\label{sec:prelim}

\subsection{Notations} For any set $X$ we denote by $2^X$, the power set of $X$.  We use the notation $O_k(\cdot)$ (resp.\ $\Omega_k(\cdot)$) to mean $F(k)\cdot O(\cdot)$ (resp.\ $F(k)\cdot \Omega(\cdot)$) for some function $F$.

\subsection{Problems and Hypotheses}

In this subsection, we formally define all the problems and hypotheses used in the paper. 

First, we define the $\ell$-\SAT problem and then define the two popular fine-grained hypotheses concerning this problem. 

\paragraph{\bf $\ell$-\SAT}
In the $\ell$-\SAT problem, we are given a \CNF formula $\varphi$ over $n$ variables $x_1,\ldots x_n$, such that each clause contains at most $\ell$ literals. Our goal is to decide if there exist an assignment to $x_1,\ldots x_n$ which satisfies $\varphi$.

In this paper, we require a fine-grained notion (of algorithms) in the complexity class \textsf{RP} and a fine-grained notion of  \emph{Reverse Unfaithful Random (RUR) reductions}~\cite{J90,MG02}. An FPT notion of such algorithms and reductions was introduced in \cite{BBEGKLMM21} and the notion of randomized fine-grained reduction was introduced in \cite{CarmosinoGIMPS16}. A promise problem $\Pi$ is a pair of languages  $(\Pi_{\text{YES}},\Pi_{\text{NO}})$ such that $\Pi_{\text{YES}}\cap\Pi_{\text{NO}}=\emptyset$.
A Monte Carlo algorithm $\mathcal{A}$ is said to be a (one-sided) randomized algorithm for a (promise) problem $\Pi$ if
the following holds:
\begin{itemize}
\item (YES) For all $x\in \Pi_{\text{YES}}$, $\Pr[\mathcal{A}(x) = 1] \ge 1/2$.
\item (NO) For all $x\in \Pi_{\text{NO}}$, $\Pr[\mathcal{A}(x) = 0] =1$.
\end{itemize}
Moreover, we say that $\mathcal{A}$ runs in time $T$ if the running time of  $\mathcal{A}$ on every randomness is upper bounded by $T$. 
\begin{hypothesis}[(Randomized) Exponential Time Hypothesis (\ETH)~\cite{IP01,IPZ01,Tovey84}] \label{hyp:eth}
There exists an $\epsilon > 0$ such that no Monte Carlo (one-sided) randomized algorithm can solve 3-\SAT on $n$ variables in time $O(2^{\epsilon n})$. Moreover, this holds even when restricted to formulae in which each variable appears in at most three clauses.
\end{hypothesis}

We will also recall a stronger hypothesis called the Strong Exponential Time Hypothesis (\SETH):

\begin{hypothesis}[(Randomized) Strong Exponential Time Hypothesis (\SETH)~\cite{IP01,IPZ01}]
For every $\varepsilon > 0$, there exists $\ell = \ell(\varepsilon) \in \NN$ such that no Monte Carlo (one-sided) randomized algorithm can solve $\ell$-\SAT  in $O(2^{(1 - \varepsilon)m})$ time where $m$ is the number of variables. Moreover, this holds even when the number of clauses is at most $c(\varepsilon) m$ where $c(\varepsilon)$ denotes a constant that depends only on $\varepsilon$.
\end{hypothesis}

In this paper,  we prove tight running time lower bounds for \intersection (to be formally defined later in this section) assuming \ETH (resp.\ \SETH)  by providing a \emph{fine-grained RUR reduction} from 3-\SAT (resp.\ $\ell$-\SAT) to \intersection, such that YES instances of 3-\SAT (resp.\ $\ell$-\SAT) map to YES instances of \intersection with high probability and NO instances of 3-\SAT (resp.\ $\ell$-\SAT) always map to NO instances of \intersection. We remark that  using standard techniques, fine-grained RUR reductions can be used to transform Monte Carlo one-sided randomized  algorithms for \intersection  to  Monte Carlo one-sided randomized 
algorithms for \SAT  (for example, see Lemma~3.7 in \cite{BBEGKLMM21}). 

Next, we recall the \maxcover problem introduced in \cite{CCKLMNT20} which turned out to be  the centerpiece of many results in parameterized inapproximability. 

\paragraph{$k$-\textsf{MaxCover} problem}
The $k$-\maxcover instance $\Gamma$ consists of a bipartite graph $G=(V\dot\cup W, E)$ such that $V$ is partitioned into $V=V_1\dot\cup \cdots \dot\cup V_k$ and $W$ is partitioned into $W=W_1\dot\cup \cdots \dot\cup W_\ell$. We sometimes refer to $V_{i}$'s and $W_j$'s as {\em left super-nodes} and {\em right super-nodes} of $\Gamma$, respectively. 

A solution to $k$-\maxcover is called a {\em labeling},
which is a subset of vertices $v_1\in V_1,\dots v_k\in V_k$.
We say that a labeling $v_1,\dots v_k$ {\em covers} a right super-node $W_{i}$, if there exists a vertex $w_{i} \in W_{i}$ which is a joint neighbor of all $v_1,\dots v_k$, i.e., $(v_j,w_i)\in E$ for every $j\in[k]$.
We denote by $\maxcover(\Gamma)$ the maximal fraction of right super-nodes that can be simultaneously covered, i.e.,
\begin{align*}
\maxcover(\Gamma) = \frac{1}{\ell} \left(\max_{\text{labeling } v_1,\dots v_k} \bigl\lvert\bigl\{i \in [\ell] \mid W_i \text{ is covered by } v_1,\dots v_k\bigr\}\bigr\rvert\right). 
\end{align*}

Given an instance $\Gamma(G,c,s)$ of the  $k$-\maxcover problem as input,  our goal is to distinguish between the two cases:
\begin{description}
	\item[Completeness] $\maxcover(\Gamma)\ge c$.   
	\item[Soundness] $\maxcover(\Gamma)\le s$.  
\end{description}

We define {\sf Unique} \maxcover to be the \maxcover problem with the following additional structure: for every labeling $S\subseteq V$  and any right super-node $W_i$, there is at most one node in $W_i$ which is a neighbor to all the nodes in $S$.

%

Next, we define the two central computational problems of attention in this paper, \intersection and its coloured variant, panchromatic \intersection.

\paragraph{\intersection problem}
The \intersection instance $\Gamma$ consists of a collection $\C$ of $n$ subsets of a universe $\U$ (typically synonymous with $[n]$) and integer parameters $c,s$ ($c>s$). 
In the \intersection problem, given input $\Gamma (\C,c,s)$, the goal is to distinguish between the two cases:
\begin{description}	
\item[Completeness] There exists $k$ sets $S_{i_1},\ldots ,S_{i_k}$ in $\C$ such that $\left|\underset{{r\in[k]}}{\cap}S_{i_r}\right|\ge c$.
	\item[Soundness] For every $k$ sets $S_{i_1},\ldots ,S_{i_k}$ in $\C$ we have $\left|\underset{{r\in[k]}}{\cap}S_{i_r}\right|\le s$.
\end{description}

\paragraph{Panchromatic \intersection problem}
The panchromatic \intersection instance $\Gamma$ consists of $k$ collections $\C_1,\ldots \C_k$ each containing $n$ subsets of a universe $\U$ and integer parameters $c,s$ ($c>s$). 
In the panchromatic \intersection problem, given input $\Gamma (\C_1,\ldots \C_k,c,s)$, the goal is to distinguish between the two cases:
\begin{description}	
\item[Completeness] There exists $k$ sets $S_{i_1},\ldots ,S_{i_k}$ in $\C_1\times\cdots \times \C_k$ such that $\left|\underset{{r\in[k]}}{\cap}S_{i_r}\right|\ge c$.
	\item[Soundness] For every $k$ sets $S_{i_1},\ldots ,S_{i_k}$ in $\C_1\times\cdots \times\C_k$ we have $\left|\underset{{r\in[k]}}{\cap}S_{i_r}\right|\le s$.
\end{description}
We define an important quantity for instances of panchromatic \intersection, which we call the \emph{monochromatic number} of $\Gamma$ and is defined to be the following quantity:
$$
\max_{\substack{X\subseteq \C_1\cup \cdots \cup \C_k\\
|X|=k}}\left|\bigcap_{S\in X}S\right|
$$

Additionally, we make the following connection between {\sf Unique} $k$-\maxcover
and panchromatic \intersection. 
\begin{proposition}
Every {\sf Unique} \maxcover instance \[\Gamma(V:=V_1\dot\cup \cdots \dot\cup V_k ,W:=W_1\dot\cup \cdots \dot\cup W_\ell,E,c,s)\] is also a panchromatic \intersection instance $\Gamma'(\C_1,\ldots ,\C_k,c',s')$ over universe $\U$ with monochromatic number $z$ where we have 
(i) $|\U|=|W|$,
 (ii) $\forall i\in[k]$, $|\C_i|=|V_i|$,
(iii) $c'=c\cdot \ell$,
(iv) $s'=s\cdot \ell$, and (v)
 $z\le |W|$.\label{prop:mctointer}
\end{proposition}
\begin{proof}
For every $w\in W$ we create a universe element $u_w\in\U$. 
For every $v\in V_i$ we create a set $S_v\in\C_i$ and we include $u_w$ in $S_v$ if there is an edge between $w$ and $v$ in $\Gamma$. 
Note that $w$ is a common neighbor of $(v_1,\ldots ,v_k)\in V_1\times \cdots V_k$ if and only if $u_w$ is in $\cap_{i\in[k]}S_{v_i}$. Furthermore note that since $\Gamma$ is an instance of {\sf Unique} $k$-\maxcover, we have that the quantity $\ell\cdot\left( \text{\maxcover}(\Gamma)\right)$ is simply the number of common neighbors of any $k$ vertices in $V$ when we pick one vertex from each $V_i$. The theorem statement then follows. 
\end{proof}

Finally, we define a contrapositive version of \intersection problem as this variant comes in handy to describe a gap creation approach in Appendix~\ref{sec:mincover}. 

\paragraph{\mincover problem}
The \mincover instance $\Gamma$ consists of a collection $\C$ of $n$ subsets of $[n]$ and integer parameters $c,s$ ($c<s$). 
In the \mincover problem, given input $\Gamma (\C,c,s)$, the goal is to distinguish between the two cases:
\begin{description}	
\item[Completeness] There exists $k$ sets $S_{i_1},\ldots ,S_{i_k}$ in $\C$ such that $\left|\underset{{r\in[k]}}{\cup}S_{i_r}\right|\le c$.
	\item[Soundness] For every $k$ sets $S_{i_1},\ldots ,S_{i_k}$ in $\C$ we have $\left|\underset{{r\in[k]}}{\cup}S_{i_r}\right|\ge s$.
\end{description}

\paragraph{Panchromatic \mincover problem}
The panchromatic \mincover instance $\Gamma$ consists of $k$ collections $\C_1,\ldots \C_k$ each containing $n$ subsets of $[n]$ and integer parameters $c,s$ ($c<s$). 
In the panchromatic \mincover problem, given input $\Gamma (\C_1,\ldots \C_k,c,s)$, the goal is to distinguish between the two cases:
\begin{description}	
\item[Completeness:] There exists $k$ sets $S_{i_1},\ldots ,S_{i_k}$ in $\C_1\times\cdots \times \C_k$ such that $\left|\underset{{r\in[k]}}{\cup}S_{i_r}\right|\le c$.
	\item[Soundness:] For every $k$ sets $S_{i_1},\ldots ,S_{i_k}$ in $\C_1\times\cdots \times\C_k$ we have $\left|\underset{{r\in[k]}}{\cup}S_{i_r}\right|\ge s$.
\end{description}

\section{Panchromatic and Threshold Graphs: Definitions and Results}
\label{sec:graphs}
Here, we define panchromatic and threshold graphs a little more carefully, and also state precisely what our constructions accomplish.

We start with \emph{panchromatic} graphs.

\begin{definition}[$(n,m,k,t,s,p)$-panchromatic graph]\label{def:pan}
 A bipartite graph $G(A,B)$ where $A$ is partitioned into $k$ parts $A_1,\ldots,A_k$ with $|A_1| = \dots = |A_k|=n$ and $|B|\le m$ satisfying the following pair of conditions.
\begin{description}	
\item[Completeness] 
For a $p$-fraction of the $k$-sets $\{a_1,a_2,....,a_k\}$ with $a_i \in A_i$ for $i \in [k]$, the number of common neighbours of $\{a_1,a_2,....,a_k\}$ in $B$ is exactly $t$, and every $k$-set $\{a_1,a_2,....,a_k\}$ with $a_i \in A_i$ for $i \in [k]$ has at most $t$ common neighbours in $B$.
\item[Soundness] For every set $X \subset A$ of size $k$ for which $A_i \cap X$ is empty for some $i\in [k]$, the number of common neighbours of $X$ in $B$ is at most $s$.
\end{description}
\end{definition}

Next, we turn to \emph{threshold} graphs.
\begin{definition}[$(n,m,k,t,s,p)$-threshold graph]\label{def:thres}
A bipartite graph $G(A,B)$ with $|A|=n$ and $|B|\le m$ satisfying the following pair of conditions.
\begin{description}	
\item[Completeness] For a $p$-fraction of $k$-sets of vertices $\{a_1,a_2,....,a_k\} \subset A$, the number of common neighbours of $\{a_1,a_2,....,a_k\}$ in $B$ is at least $t$.
\item[Soundness] For every $(k+1)$-set of vertices $\{a_1,a_2,....,a_{k+1}\}$ in $A$, the number of common neighbours of $\{a_1,a_2,....,a_{k+1}\}$ in $B$ is at most $s$.
\end{description}
\end{definition}

We show that both types of graphs may be constructed with reasonable dependencies between the various parameters involved. Both constructions are easy to describe, with the edge sets of the graphs in question coming from the varieties cut out by (carefully chosen) random polynomials; the analysis of these constructions is far from trivial however, and relies on some amount of machinery from algebraic geometry.

For panchromatic graphs, we have the following result which, in particular, ensures that such graphs exist.


\begin{theorem}\label{thm:pan_cons}
For each $k \in \NN$ and any integer $\lambda > 1$, there is a strictly increasing sequence $\{n_i \in \NN\}_{i\in\NN}$ such that for every $i\in\mathbb{N}$, there exists a distribution $\D_{k,\lambda,n_i}$ over bipartite graphs on $(k+1)n_i$ vertices with the following properties.
\begin{enumerate}
    \item A graph can be sampled from $\D_{k,\lambda,n_i}$ in ${O}_k(n_i^2)$ time using $O_k(n_i\log n_i)$ random coins.
    \item For $G\sim \D_{k,\lambda,n_i}$, writing $D =  \lambda(k^2 + 2) $, we have
    \[{\pr}\left(G\text{ is a }(n_i, n_i, k, D^k, D^{k}/\lambda, (4(D^k)!)^{-1})\text{-panchromatic graph} \right)\ge (4(D^k)!)^{-1}.\] 
\end{enumerate}
Moreover, for every $n\in\NN$, there exists $i\in\NN$ such that $n\le n_i\le 2^k\cdot n$.
\end{theorem}

For threshold graphs, we have the following analogous result, which again, in particular, ensures that such graphs exist.

\begin{theorem}\label{thm:thresh_cons}
For each $k \in \NN$, there is a strictly increasing sequence $\{n_i \in \NN\}_{i\in\NN}$ such that for every $i\in\mathbb{N}$, there exists a distribution $\D_{k,n_i}$ over bipartite graphs on $2n_i$ vertices with the following properties.
\begin{enumerate}
    \item A graph can be sampled from $\D_{k,n_i}$ in ${O}_k(n_i^2)$ time using $O_k(n_i\log n_i)$ random coins.
    \item For $G\sim \D_{k,n_i}$, writing $d = (k+1)^2 + 1$, we have
    \[\pr \left(G\text{ is a }(n_i, n_i, k, n_i^{1/(k+1)}/2, d^{k+1}, 1)\text{-threshold graph}\right) \ge 1-o(1).\]
\end{enumerate}
Moreover, for every $n\in\NN$, there exists $i\in\NN$ such that $n\le n_i\le 2^k\cdot n$.
\end{theorem}

\section{Zero sets of Random Polynomials}\label{sec:tools}

The aim of this section is to collect together the requisite tools from algebraic geometry that we require to prove Theorems~\ref{thm:pan_cons} and~\ref{thm:thresh_cons}. While we have attempted to keep the presentation self-contained for the most part, some of the arguments (unavoidably) assume some familiarity with algebraic geometry; for more background, we refer the reader to~\citep{sha, Ful84}.

A variety over an algebraically closed field $\oFF$ is a set of the form
\[V = \{x \in \oFF^k : f_1(x) = \dots = f_t(x) = 0\}\]
for some collection of polynomials $f_1, \dots, f_t\colon \oFF^k \rightarrow \oFF$; when we wish to make these polynomials explicit, we write $V(f_1, \dots, f_t)$ for $V$.  A variety is said to be \emph{irreducible} if it cannot be written as the union of two proper subvarieties. The \emph{dimension} $\dim V$ of a variety $V$ is then the maximum integer $d$ such that there exists a chain of irreducible subvarieties of $V$ of the form
\[\emptyset \subsetneq V_0 \subsetneq V_1 \subsetneq V_2 \subsetneq \dots \subsetneq V_d \subset V,\]
where $V_0$ consists of a single point. The \emph{degree} of an irreducible variety of dimension $d$ is the number of intersection points of the variety with $d$ hyperplanes in general position, and for an arbitrary variety $V$, we define its degree $\deg V$ to be the sum of the degrees of its irreducible components.

We need B\'ezout's theorem in the following form; for a proof, see~\citep[p.~223, Example 12.3.1]{Ful84}, for example.
\begin{lemma}\label{bezout}
	For a collection of polynomials $f_1, \dotsc, f_k \colon \oFF^k \rightarrow \oFF$, if the variety
	\[V = \{x \in \oFF^k : f_1(x) = \dots = f_k(x) = 0\}\]
	has $\dim V = 0$, then 
	\[|V| \leq \prod_{i=1}^k \deg(f_i).\]
	Moreover, for a collection of polynomials $f_1, \dotsc, f_t\colon \oFF^k \rightarrow \oFF$, the variety
	\[V = \{x \in \oFF^k : f_1(x) = \dots = f_t(x) = 0\}\]
	has at most $\prod_{i=1}^t \deg(f_i)$ irreducible components.
\end{lemma}

In what follows, we let $q$ be a prime power and work with polynomials over $\F_q$, where $\mathbb{F}_q$ is the finite field of order $q$. All varieties below are over $\A$,  where $ \A = \oFF_q$ is the algebraic closure of $\F_q$, unless explicitly specified otherwise.  We let $\F_q[X_1, \dots, X_k]_{\leq d}$ be the subset of $\F_q[X_1, \dots, X_k]$ of polynomials in $k$ variables of degree at most $d$, i.e., the set of linear combinations over $\mathbb{F}_q$ of monomials of the form $X_1^{a_1} \dots X_k^{a_k}$ with $\sum_{i=1}^k a_i \leq d$. Let us note that one may sample a uniformly random element of $\F_q[X_1, \dots, X_k]_{\leq d}$ by taking the coefficients of the monomials above to be independent random elements of $\mathbb{F}_q$.

The first lemma we state estimates the probability of a randomly chosen polynomial passing through each of $m$ distinct points; see~\citep{bukh, conlon} for similar statements.

\begin{lemma}\label{lem:vanish}
  Suppose that $q > \binom{m}{2}$ and $d \geq m - 1$. Let $f$ be a uniformly random $k$-variate polynomial chosen from $\F_q[X_1, \dots, X_k]_{\leq d}$.
  \begin{enumerate}
    \item If $x_1, \dots, x_m$ are $m$ distinct points in $\F_q^k$, then
	\[ \pr \left(f(x_i) = 0 \mbox{ for all } i = 1, \dots, m\right) = q^{-m}. \]
    \item If $x_1, \dots, x_m$ are $m$ distinct points in $\oFF_q^k$, then
      \[ \pr \left(f(x_i) = 0 \mbox{ for all } i = 1, \dots, m\right) \leq q^{-m}. \]
  \end{enumerate}
\end{lemma}
\begin{proof}
    We prove the first statement below, and later outline the proof of the second statement.
  
	Let $x_i = (x_{i,1}, \dots, x_{i,k})$ for each $i = 1, \dots, m$. We choose elements $a_2, \dots, a_k  \in \mathbb{F}_q$ such that $x_{i,1} + \sum_{j=2}^k a_j x_{i, j}$ is distinct for all $i = 1, \dots, m$. To see that this is possible, note that there are exactly $\binom{m}{2}$ equations 
	\[x_{i,1} + \sum_{j=2}^k a_j x_{i, j} = x_{i',1} + \sum_{j=2}^k a_j x_{i', j},\] 
	each with at most $q^{k-2}$ solutions $(a_2, \dots, a_k)$. Therefore, since the total number of choices for $(a_2, \dots, a_k)$ is $q^{k-1}$ and $q^{k-1} > q^{k-2} \binom{m}{2}$, we can make an appropriate choice.
	
	We now consider $\F_q[Z_1, \dots, Z_k]_{\leq d}$, the set of polynomials of degree at most $d$ in the variables $Z_1, \dots, Z_k$, where $Z_1 = X_1 + \sum_{j=2}^k a_j X_j$ and $Z_j = X_j$ for all $2 \leq j \leq k$. Since this change of variables is an invertible linear map, $\F_q[Z_1, \dots, Z_k]_{\leq d}$ is identical to $\F_q[X_1, \dots, X_k]_{\leq d}$. It will therefore suffice to show that a randomly chosen polynomial from $\F_q[Z_1, \dots, Z_k]_{\leq d}$ passes through all of the points $z_1, \dots, z_m$ corresponding to $x_1, \dots, x_m$ with probability exactly $q^{-m}$. For this, we will use the fact that, by our choice above, $z_{i,1} \neq z_{i',1}$ for any $1 \leq i < i' \leq m$.
	
	For any $f$ in $\F_q[Z_1, \dots, Z_k]_{\leq d}$, we may write $f = g + h$, where $h$ contains all monomials of the form $Z_1^j$ for $j = 0, 1, \dots, m-1$ and $g$ contains all other monomials. For any fixed choice of $g$, there is, by Lagrange interpolation, exactly one choice of $h$ with coefficients in $\F_{q}$ such that $f(z_i) = 0$ for all $i = 1, \dots, m$, namely, the unique polynomial of degree at most $m-1$ which takes the value $-g(z_i)$ at $z_{i,1}$ for all $i = 1, 2, \dots, m$, where uniqueness follows from the fact that the $z_{i,1}$ are distinct.
        Since this is out of a total of $q^m$ possibilities, we see that the probability of $f$ passing through all of the $z_i$ is exactly $q^{-m}$, as required.
        
    For the second statement, we may argue identically, now working over $\kbar$ and noting that the unique polynomial of degree at most $m-1$ which takes the value $-g(z_i)$ at $z_{i,1}$ for all $i = 1, 2, \dots, m$ may now have coefficients in $\kbar$ as opposed to $\k$, whence we get an inequality as opposed to the equality in the first statement.
\end{proof}

The next result we prove allows us to upper bound the size of the $\F_q$-variety cut out by multiple random polynomials.

\begin{theorem}\label{thm:toomany}
Fix $t,k \in \NN$ with $t \le k$, and fix positive integers $d_1,\dots,d_t\in\NN$. Independently for each $i\in [t]$, sample $f_i$ from $\k[X_1, \dots, X_k]_{\leq d_i}$ uniformly at random.
Then
	\begin{equation} \label{eq:vardim}
		\pr\left(\dim V(f_1,\dots,f_t)>k-t\right)\leq C_tq^{-\min(d_1,\dots,d_t)}
	\end{equation}
for some constant $C_t = C_t(d_1,\dots,d_k) > 0$. In particular, if $t=k$, then 
\[ 		\pr\left(\left|V(f_1,\dots,f_k) \cap \F_q^k\right| > \prod_{i=1}^k d_i \right)\leq C q^{-\min(d_1,\dots,d_k)} \]
for some constant $C = C(d_1,\dots,d_k) > 0$.
\end{theorem}
\begin{proof}
For terminology not defined here, and standard facts about dimension that we call upon without proof, see the first and the sixth chapter of~\citep{sha}.

	To establish \eqref{eq:vardim} it suffices show that
	\begin{equation}\label{eq:condvardim}
		\pr\left(\dim V(f_1,\dots,f_{t-1},f_t)>k-t \mid \dim V(f_1,\dots,f_{t-1})=k-t+1\right)\leq q^{-d_t}\prod_{i=1}^{t-1} d_i
	\end{equation}
	since~\eqref{eq:vardim} follows from \eqref{eq:condvardim} by induction on $t$.
	
	Now, sample polynomials $f_1,\dots,f_{t-1}$, and assume that the variety $U = V(f_1,\dots,f_{t-1})$ is of dimension $d-t+1$. By Lemma~\ref{bezout}, $U$ has at most $d_1\cdots d_{t-1}$ components, which we name $U_1,\dots,U_m$. Note that since $\dim U_i\leq \dim U=d-t+1$, and $U_i$ is intersection of $t-1$ hypersurfaces, each $U_i$ is of dimension exactly $d-t+1$. For each $U_i$, pick $d_t$ distinct points $x_{i,1},\dots,x_{i,d_t}$ on $U_i$.

	Since $f_t$ is a random polynomial of degree $d_t$, from Lemma~\ref{lem:vanish} we infer that
	\[
	\pr\left(U_i\subset V(f_t)\right)\leq \pr\left( f_t(x_i,j)=0 \text{ for all } j = 1, \dots, d_t \right) \leq q^{-d_t}
	\]
	for each $1 \le i \le m$.
	Hence, by the union bound
	\[
	\pr\left(\dim V(f_1,\dots,f_{t-1},f_t)>k-t\right)\leq \sum_{i=1}^m \pr\left(U_i\subset V(f_t)\right)\leq q^{-d_t}\prod_{i=1}^{t-1} d_i.
	\]
	proving \eqref{eq:condvardim}, and hence \eqref{eq:vardim}.  \smallskip
	
	If $t=k$, then
	\begin{align*}
		 		\pr\left(\left|V(f_1,\dots,f_k) \cap \F_q^k\right| > \prod_{i=1}^k d_i \right) &\leq 
	\pr \left( \abs{V(f_1,\dots,f_k)} > \prod_{i=1}^k d_i \right) \\
	&\le \pr(\dim V(f_1,\dots,f_k)> 0)\\
	&\le C_kq^{-\min(d_1,\dots,d_k)},
	\end{align*}
	 where the first inequality is trivial, the second is a consequence of Lemma~\ref{bezout}, i.e., B\'ezout's theorem, and the third is just~\eqref{eq:vardim} for $t = k$.
\end{proof}

Finally, we need a way to lower bound the size of the $\F_q$-variety cut out by multiple random polynomials, and the following result gives us what we need.
While the arguments thus far have been mostly elementary, this result is more involved.

\begin{theorem}\label{thm:exactly}
Fix positive integers $k, d_1, \dots,d_k \in \NN$. Independently for each $i\in [k]$, sample $f_i$ from $\k[X_1, \dotsc, X_k]_{\leq d_i}$ uniformly at random.
Then
\[ 		\pr\left(\left|V(f_1,\dots,f_k) \cap \F_q^k\right| = \prod_{i=1}^k d_i \right) \ge \frac{1 - cq^{-1/2}}{\left(\prod_{i=1}^k d_i\right)!} \]
for some constant $c = c(d_1, \dots,d_k) > 0$.
\end{theorem}

\begin{proof}
For terminology not defined here, and standard results that we quote without proof, see the first three chapters of~\citep{sha}.

We set $r_i =\binom{k+d_i}{k}$ for $1 \le i \le k$, write $\vr = (r_1,\dots,r_k)$ and $\abs{\vr}$ for $r_1+\dotsb+r_k$. For $1 \le i \le k$, we identify $\A^{r_i}$ with $\A[X]_{\leq d_i}$, i.e., the space of polynomials in $k$ variables of degree at most $d_i$ with coefficients in $\A$. 
For brevity, we write $\A^{\vr}$ in place of $\A^{r_1}\times\dots \times \A^{r_{k}}$ (and $\F_q^{\vr}$ in place of $\F_q^{r_1}\times\dots \times \F_q^{r_{k}}$), and to distinguish the space where we evaluate our polynomials from these spaces of polynomials themselves, we set $Y = \A^k$.

Also, for $\vf=(f_1,\dots,f_k)\in \A^{\vr}$, we abbreviate the variety $V(f_1, \dotsc, f_k) \subset Y$ by $V({\vf})$. Now, set $t = d_1 \cdots d_k$ and call $\vf \in \k^{\vr}$ \emph{good} if the variety $V({\vf})$ is zero-dimensional and has $t$ distinct points that are defined over $\k$. In this language, note that we are trying to show, for large $q$, that roughly $1/t!$ of all the points in $\k^{\vr}$ are good. To this end, we set
\[
W = \{(\vf,y_1,\dots,y_t) \in \A^{\vr} \times Y^t : y_j\in V({\vf}) \text{ for all } j=1,\dots,t\},\]
and deduce the result from the following claim. 

\begin{claim}\label{main-thm}
	Suppose that $(\vf^*,\vy^*)$ is a simple point of $W$ such that $\vf^*$ is good and the coordinates of $\vy^* = (y_1^*,\dots,y_t^*)$ are all distinct, and that for generic $\vf$, the variety $V({\vf})$ is zero-dimensional of degree~$t$. Then there are at least 
	\[\frac{1-cq^{-1/2}}{t!}q^{\abs{\vr}}\] good points in  $\k^{\vr}$, for some constant $c = c(d_1, \dotsc,, d_k) > 0$. 
\end{claim}
\begin{proof}
	Since $(\vf^*,\vy^*)$ is simple, the irreducible component of $W$ containing it is unique. 
	Let $W_1$ be the irreducible component of $W$ containing $(\vf^*,\vy^*)$ and note that $\dim W_1=\dim W$. 
	Since the variety $V({\vf})$ is generically zero-dimensional of degree~$t$, the fibres $W_{\vf} = \{\vy \in Y^t : (\vf,\vy)\in W\}$ of $W$ are generically finite, whence we get $\dim W_1 = \dim W=\abs{\vr}$.
	
	Let $\{W_1,\dots,W_m\}$ be the orbit of $W_1$ under the action of the Frobenius endomorphism. Since $W$ is defined over $\k$, and hence invariant under this action,
	each such $W_i$ is an irreducible component of $W$. Note that $(\vf^*,\vy^*)\in W_i$ for each $i\in [m]$, so if $m>1$, this contradicts the uniqueness of the component containing $(\vf^*,\vy^*)$. Thus, $m=1$, i.e., $W_1$ is defined over $\k$.
		
	Since $(\vf^*,\vy^*)\in W_1$, the variety $W_1$ is not contained in 
	\[U = \bigcup_{i\neq j}\{(\vf,\vy) : y_i=y_j\}.\]
	Hence, $W_1\cap H$ is a proper subvariety of $W_1$, and therefore contains $O_{\deg W_1}(q^{\abs{\vr}-1})$ points by the Schwartz--Zippel lemma for varieties~\citep[Lemma~14]{bukh_tsimerman}. Since $W_1$ is defined over $\k$ and is irreducible over $\A$, the Lang--Weil estimate~\citep{LW54} implies that $W_1$ contains at least
	\[q^{\dim W_1}\left(1-O_{\deg W_1}( q^{-1/2})\right)\]
	points defined over $\k$. Hence, $W_1\setminus H$ contains at least
	 	\[q^{\abs{\vr}}\left(1-O_{\deg W_1}( q^{-1/2}) - O_{\deg W_1}( q^{-1})\right) = q^{\abs{\vr}}\left(1-O_{\deg W_1}( q^{-1/2})\right)\] 
	points defined over $\k$ as well. Since each good point $\vf$ corresponds to exactly $t!$ points of $W_1\setminus H$ defined over $\k$, the result follows.
\end{proof}
To finish, it remains to show that the simplicity and genericity hypotheses in Claim~\ref{main-thm} are satisfied. 

For $1 \le i \le k$,  pick an arbitrary set $A_i \subset \k$ of size $d_i$. Define $\vf^* = (f^*_1, \dots,f^*_k)$ by setting $f_i^*=\prod_{a\in A_i}(X_i-a)$ for $1 \le i \le k$ and let $\vy^*$ be the vector of length $d_1 \cdots d_k$ whose coordinates are all the elements of $A_1 \times \dots \times A_k$.

To prove that $(\vf^*,\vy^*)$ is simple, consider the tangent space of $W$ at $(\vf^*,\vy^*)$, which we denote $T_{*} W$.
An element $(\delta \vf,\delta \vy)\in \A^{\vr} \times Y^{t}$ is in $T_{*} W$ if it is a solution to the system of equations
\[
\delta f_i(y_j^*)+\frac{\partial f_i}{\partial x_i}(y_j^*)(\delta y_j)_i =0 
\]
for all $i\in [k]$ and $j\in [t]$.
From these equations, it is clear that for every $\delta \vf\in \A^{\vr}$ there is a unique $\delta \vy$ such that $(\delta \vf,\delta \vy)$ is in the tangent space. Hence $\dim T_{*} W=\dim \A^{\vr}=\dim W$, so it follows that $(\vf^*,\vy^*)$ is simple.

Next, the statement that for generic $\vf$, the variety $V({\vf})$ (is zero-dimensional and) has at most $t=d_1 \cdots d_k$ points is the generalized B\'ezout's theorem. The construction of $(\vf^*,\vy^*)$ above shows that $V({\vf})$ generically has at least $t$ points as well.

We have established the hypotheses under which Claim~\ref{main-thm} applies; the result follows.
\end{proof}

\section{Constructions of Panchromatic Graphs and Threshold Graphs}\label{sec:cons}

First, we give the construction of panchromatic graphs using random polynomials.

\begin{proof}[Proof of Theorem~\ref{thm:pan_cons}]
Let $q$ be a prime power, and let $\F_q$ be the finite field of order $q$. We shall assume that $k\in \NN$ and $\lambda > 1$ are fixed, and that $q$ is sufficiently large as a function of $k$. Finally, let us fix $d = k^2+2$, $D =  \lambda d $ and $n=q^{k}$. In the rest of the proof, all asymptotic notation will be in the limit of $q \to \infty$.

We shall construct a panchromatic graph between two sets $A$ and $B$ as follows. First, choose polynomials $w_1, \dotsc, w_k \in \F_q [X_1, \dotsc, X_k]_{\le D}$  independently and uniformly at random. Next, for $i \in [k]$, let $A_i$ be a set of $n$ vertices each associated with a polynomial $w_i + p$, where $p \in  \F_q [X_1,\dotsc, X_k]_{\le d}$ is chosen uniformly at random and independently for each vertex; note here that the distribution of the resulting polynomial $w_i + p$ is also uniform on $\F_q [X_1, \dotsc, X_k]_{\le D}$. Let $A$ be the disjoint union $\dot\cup_{i=1}^k A_i$, and set $B = \F_q^k$, so that $|A| = k q^k$ and $|B| = q^k$. Finally, let $G$ be the (random) graph between $A$ and $B$ where a polynomial $f \in A$ is joined to a point $x \in B$ if $f(x) = 0$. We shall show that $G$ has the requisite properties with probability at least ${(4(D^k)!)^{-1}}$.

First, we count the number of $k$-sets $U = \{f_1,f_2,\dotsc,f_k\}$ with $f_i \in A_i$ for which the size of the common neighbourhood $N(U)$ in $G$ exceeds $D^k$. For such a set~$U$, observe that $N(U)$ is the set of $\F_q$-solutions of $k$ polynomials from $\F_q[X_1, \dotsc, X_{k}]_{\le D}$ chosen independently and uniformly at random, so by Theorem~\ref{thm:toomany}, we have
\[ \pr(|N(U)| > D^{k}) = O(q^{-D}).\]
Writing $B_1$ for the number of such $k$-sets, we get 
\begin{equation}\label{eq:pb1}
	\E[B_1] = O\left(n^kq^{-D}\right)= O\left(q^{k^2}q^{-\lambda(k^2+2)}\right) = O(q^{-2}) \le 1/q.
\end{equation}

Next, we count the number of $k$-sets $U = \{f_1,f_2,....,f_k\}$ with $f_i \in A_i$ for $i \in [k]$ for which size of the common neighbourhood $N(U)$ in $G$ is exactly $D^k$. As above, for such a set $U$, observe that $|N(U)|$ is distributed as the number of $\F_q$-solutions of $k$ polynomials from $\F_q[X_1, \dots X_{k}]_{\le D}$ chosen independently and uniformly at random, so by Theorem~\ref{thm:exactly}, we have
\[ \pr(|N(U)| = D^{k}) \ge (2(D^k)!)^{-1}.\]
Writing $B_2$ for the number of such $k$-sets, we get 
\begin{equation}\label{eq:pb2}
	\E[B_2] \ge n^k(2(D^k)!)^{-1}.
\end{equation}

Finally, we count the number of $k$-sets $U \subset A$ with $A_i \cap U$ being empty for some $i\in [k]$ for which the size of the common neighbourhood $N(U)$ in $G$ exceeds $dD^{k-1} = D^k / \lambda$. For such a set $U$, observe that $|N(U)|$ is distributed as the number of $\F_q$-solutions of a collection of $k$ random polynomials. To understand the distribution of this random collection of polynomials, for each $i \in [k]$ for which $U \cap A_i \ne \emptyset$, we pick one element $U \cap A_i$ and subtract that from every other element of $U \cap A_i$; observe that by doing so, we get a set $\{g_1, \dotsc, g_k\}$ of independent random polynomials, each uniform over either $\F_q[X_1, \dots X_{k}]_{\le d}$ or $\F_q[X_1, \dotsc, X_{k}]_{\le D}$, and at least one of which is uniform over $\F_q[X_1, \dots X_{k}]_{\le d}$. Since $|N(U)|$ is then number of $\F_q$-solutions of  $\{g_1, \dotsc, g_k\}$, we deduce from Theorem~\ref{thm:toomany} that
\[ \pr(|N(U)| > dD^{k-1} ) = O(q^{-d}).\]
Writing $B_3$ for the number of such $k$-sets, we get 
\begin{equation}\label{eq:pb3}
	\E[B_3] = O\left((kn)^kq^{-d}\right)= O\left(q^{k^2}q^{-k^2-2}\right) = O(q^{-2}) \le 1/q.
\end{equation}

We combine~\eqref{eq:pb1},~\eqref{eq:pb2} and~\eqref{eq:pb3} as follows. Clearly, $\E[B_1 + B_3] = o(1)$, so by Markov's inequality, both $B_1$ and $B_2$ are zero with probability $1-o(1)$. Finally, since $B_2$ is trivially at most $n^k$ and $\E[B_2] \ge n^k(2(D^k)!)^{-1}$, it is easily checked that
\[ \pr\left(B_2\ge n^k(4(D^k)!)^{-1} \right) \ge {(2(D^k)!)^{-1}}.\]
By the union bound, we see that $G$ is a $(n, n, k, D^k, D^{k}/\lambda, (4(D^k)!)^{-1})$-panchromatic graph with probability at least ${(4(D^k)!)^{-1}}$, completing the proof.
\end{proof}

Next, we give the construction of threshold graphs, once again using random polynomials.

\begin{proof}[Proof of Theorem~\ref{thm:thresh_cons}]
As before, let $q$ be a prime power, and let $\F_q$ be the finite field of order $q$. We shall assume that $k\in \NN$ is fixed, and that $q$ is sufficiently large as a function of $k$. Let $d=(k+1)^2 + 1$ and $n=q^{k+1}$. We shall construct a threshold graph between two sets $A$ and $B$ both of size $q^{k + 1}$. In the rest of the proof, all asymptotic notation will be in the limit of $q \to \infty$.

We construct $A$ by sampling $q^{k+1}$ random polynomials from $\F_q[X_1, \dotsc, X_{k+1}]_{ \le d}$ uniformly and independently, set $B = \F_q^{k+1}$, and define a (random) bipartite graph $G$ between $A$ and $B$ by joining $f \in A$ to $x \in B$ if $f(x) = 0$. We shall show that $G$ has the requisite properties with probability $1-o(1)$.

First, we consider the soundness properties of $G$. Fix a set $U\subset A$ of size $k+1$. The size of its common neighbourhood $N(U)$ in $G$ is distributed as the number of $\F_q$-solutions of $k+1$ polynomials from $\F_q[X_1, \dots X_{k+1}]_{\le d}$ chosen independently and uniformly at random, so by Theorem~\ref{thm:toomany}, we have
\[ \pr(|N(U)| > d^{k+1}) = O(q^{-d}).\]
Call a set of $k+1$ vertices of $G$ \emph{bad} if their common neighbourhood has more than $d^{k+1}$ vertices. The number $B_1$ of bad $(k+1)$-sets then satisfies 
\begin{equation}\label{eq:eb1}
	\E[B_1] = O\left(\binom{n}{k+1}q^{-d}\right)= O\left(\binom{q^{k+1}}{k+1}q^{-(k+1)^2 -1}\right) = O(q^{-1}) = o(1).
\end{equation}

Next, we turn to the completeness properties of $G$. Fix a set $U\subset A$ of size $k$. For $v\in B$, put $I(v)=1$ if $f(v)=0$ for all $f\in U$, and $I(v)=0$ if $f(v)\neq 0$ for some $f\in U$. For $1 \le m \le d$ and distinct $v_1,\dots, v_m \in B$, we have
\begin{align*}
\pr\left(I(v_1)\cdots I(v_m) = 1\right) = \prod_{f \in U} \pr\left(f(v_j) = 0 \text{ for all } j = 1, \dots, m\right)= q^{-mk},
\end{align*}
where the first equality is by independence, and the second is by Lemma~\ref{lem:vanish}. Small moments of the random variable $Z = \abs{N(U)}$ are now easily computed: for $1 \le m \le d$, we have
\begin{align}\label{nbrhdmoment}
  \E\left[ Z^m \right] &= \E\left[ \left(\sum_{v\in B} I(v) \right)^m \right] \nonumber \\
                               &= \E\left[ \sum_{v_1,\dots,v_m\in B} I(v_1)\cdots I(v_m)\right] \nonumber \\
                               &= \sum_{v_1,\dots,v_m\in B} \E[I(v_1)\cdots I(v_m)] \nonumber \\
                               &= \sum_{r=1}^m \binom{q^{k+1}}{r}M_{r,m} q^{-rk},
\end{align}
where $M_{r,m}$ is the number of surjective functions from an $m$-element set onto an $r$-element set. Combining~\eqref{nbrhdmoment} and some standard identities for the Stirling numbers of the second kind, we get that
\[
\E\left[ (Z - \E[Z])^d \right] = O(q) \text{ and } \E[Z] = q,
\]
whence it follows that 
 \[ \pr(Z < q/2 ) \le \pr(\abs{Z - \E[Z]}< q/2 ) \le \frac{\E\left[ (Z - \E[Z])^d \right] }{(q/2)^d} = O\left(q^{1-d}\right).
 \]
Call a set of $k$ vertices of $G$ \emph{bad} if their common neighbourhood has fewer than $q/2$ vertices. The number $B_2$ of bad $k$-sets then satisfies 
\begin{equation}\label{eq:eb2}
  \E[B_2] = O\left(\binom{n}{k} q^{1-d} \right) = O\left(\binom{q^{k+1}}{k}q^{-(k+1)^2 }\right) = O(q^{-1 - k}) = o(1).
\end{equation}

Combining~\eqref{eq:eb1} and~\eqref{eq:eb2}, we see that
\[\E[B_1 + B_2] = o(1); \]
it follows from Markov's inequality that $B_1 + B_2 = 0$ (and hence $B_1 = B_2 = 0$) with probability $1-o(1)$, so $G$ is a $(q^{k+1},q^{k+1},k,q/2, d^{k+1},1)$-threshold graph with probability $1-o(1)$, completing the proof.
\end{proof}

A quantitatively weaker version of Theorem~\ref{thm:thresh_cons} can alternately be proved utilising less randomness by building a bipartite graph between two copies of $\F_q^{k+1}$ by choosing a single random polynomial $f$ in $2k+2$ variables of degree $2k^2$ and joining pairs of points $x,y \in \F_q^{k+1}$ for which $f(x,y) = 0$; however, the analysis of this construction relies on more machinery, and furthermore, yields ineffective parameter dependencies.

\section{Conditional Time Lower Bounds for \protect\intersection}
\label{sec:inter}

In this section we prove the formal versions of Theorems~\ref{thm:introETH}~and~\ref{thm:introSETH} in Sections~\ref{sec:ETH}~and~\ref{sec:SETH} respectively. But first, we describe  in Section~\ref{sec:PGC}, the \PGC framework. 

\subsection{Panchromatic Graph Composition}\label{sec:PGC}

Given a panchromatic problem and a panchromatic graph, we would like to compose them in some way such that we obtain a \emph{monochromatic} version of the panchromatic problem having the property that every optimal solution of the monochromatic version can be traced back to  an optimal solution of the panchromatic version. When we  say the \PGC technique, we  use it as an umbrella name for this composition operation. Typically the composition would be a product operation as is the case below for the \intersection problem. 

\begin{theorem}[Panchromatic Graph Composition]\label{thm:PGC}
  There is an algorithm that given as input
  \begin{enumerate}
      \item an instance $\Gamma(\C_1,\ldots,\C_k,c,s)$ of panchromatic \intersection over universe $\U$ with monochromatic number $z$, and
      \item an (n,m,k,t,w,p)-panchromatic graph $H(A:=(A_1\dot\cup\cdots \dot\cup A_k),B)$,
  \end{enumerate}
  then outputs an instance $\Gamma'(\C',ct,\max(st,zw))$ of \intersection over universe $\U'$ such that the following hold:
  \begin{description}
  \item[Size] $|\C'|=|\C_1|+\cdots+ |\C_k|$ and $|\U'|=|\U|\cdot |B|$.
\item[Completeness] If there exists  a $k$ tuple of sets $(S_{i_1},\ldots ,S_{i_k})$ in $\C_1\times \cdots \times \C_k$ such that \[\left|\underset{{r\in[k]}}{\bigcap}S_{i_r}\right|\ge c,\] then with probability $p$ there exists $k$ sets $S_{i_1}',\ldots ,S_{i_k}'$ in $\C'$ such that \[\left|\underset{{r\in[k]}}{\bigcap}S_{i_r}'\right|\ge ct.\]
	\item[Soundness] If for every $k$ tuple of sets $(S_{i_1},\ldots ,S_{i_k})$ in $\C_1\times \cdots\times \C_k$ we have \[\left|\underset{{r\in[k]}}{\bigcap}S_{i_r}\right|\le s,\] then for every $k$ sets $S_{i_1}',\ldots ,S_{i_k}'$ in $\C'$ we have \[\left|\underset{{r\in[k]}}{\bigcap}S_{i_r}'\right|\le \max(st,zw).\]
	\item[Running Time] The reduction runs in $\tilde{O}(|\C'|\cdot |\U'|)$ time.
  \end{description}\end{theorem}
\begin{proof}
We define $\U':=\U\times B$. For every $r\in[k]$, let $\pi_r\colon\C_r\to A_r$ be a uniformly random one-to-one mapping. Moreover, for every $r\in[k]$, let $\zeta_r:\C_r\to 2^{\U'}$ be a function which maps a set in $\C_r$ to a subset of $\U'$ in $\C'$ in the following way: 
For every $S\in \C_r$, we include $\zeta_r(S)$ in $\C'$, where $(u,b)\in \U\times B$ is contained in $\zeta_r(S)$ if and only if   $u\in S$ and $(\pi_r(S),b)\in E(H)$.

Let us suppose that there exists  a $k$ tuple of sets $(S_{i_1},\ldots ,S_{i_k})$ in $\C_1\times \cdots \times \C_k$ such that \[\left|\underset{{r\in[k]}}{\bigcap}S_{i_r}\right|\ge c,\]  then consider the $k$-tuple of vertices $(\pi_1(S_{i_1}),\ldots ,\pi_k(S_{i_k}))$ in $A_1\times \cdots \times A_k$. Since $\pi_1,\ldots ,\pi_k$ were picked uniformly and independently at random, the aforementioned $k$-tuple of vertices in $A$ are $k$ uniform random vertices and thus from the completeness of the panchromatic graph, we have that with probability $p$ there exists a set of  $t$ vertices in $B$, denoted by $B'$, which are all common neighbors of $(\pi_1(S_{i_1}),\ldots ,\pi_k(S_{i_k}))$. Let $u\in \underset{{r\in[k]}}{\bigcap}S_{i_r}$ and $b\in B'$. It follows that $(u,b)\in \zeta_r(S_{i_r})$. In other words, we have:  
\[\left|\underset{{r\in[k]}}{\bigcap}\zeta_r(S_{i_r})\right|\ge c\cdot |B'|\ge ct.\]

On the other hand let us suppose that for every $k$ tuple of sets $(S_{i_1},\ldots ,S_{i_k})$ in $\C_1\times \cdots\times  \C_k$ we have \[\left|\underset{{r\in[k]}}{\bigcap}S_{i_r}\right|\le s.\] 
For the sake of contradiction, let there be $k$ sets $S_{i_1}',\ldots ,S_{i_k}'$ in $\C'$ such that  \[\left|\underset{{r\in[k]}}{\bigcap}S_{i_r}'\right|> \max(st,zw).\]
By construction of $\C'$, we have that for every $r\in[k]$, there exists $\ell_r\in[k]$  and $S_{i_r}\in \C_{\ell_r}$ such that  such that $\zeta_{\ell_r}(S_{i_r})=S_{i_r}'$.
Let $D:=\{\ell_r\mid r\in [k]\}$. Suppose that $|D|=k$, i.e., for every distinct $r_1,r_2\in[k]$ we have that $S_{i_{r_1}}$ and $S_{i_{r_2}}$ are both not in the same collection $\C_r$ (for some $r\in[k]$). Without loss of generality, we will assume $\ell_r=r$. Consider the $k$-tuple of vertices $(\pi_1(S_{i_1}),\ldots ,\pi_k(S_{i_k}))$ in $A_1\times \cdots \times A_k$. From the completeness of the panchromatic graph, we have that the set of common neighbors of $(\pi_1(S_{i_1}),\ldots ,\pi_k(S_{i_k}))$ in $B$,  denoted by $B'$,  is of size at most $t$. Thus, we have the following contradiction: 
		\[\left|\underset{{r\in[k]}}{\bigcap}S_{i_r}'\right|\le \left|\underset{{r\in[k]}}{\bigcap}S_{i_r}\right|\cdot |B'|\le st.\]
		
		Next, we suppose that $|D|<k$. Without loss of generality, we assume that $\ell_1=\ell_2$. 
Let $X:=\{\pi_{\ell_r}(S_{i_r})\mid r\in[k]\}\subseteq A$. By the soundness of the panchromatic graph, we have that the set of common neighbors of $X$ in $B$, denoted by $B'$ is at most size $w$. Thus, we have the following contradiction: 
		\[\left|\underset{{r\in[k]}}{\bigcap}S_{i_r}'\right|\le \left|\underset{{r\in[k]}}{\bigcap}S_{i_r}\right|\cdot |B'|\le zw,\]
		where $z$ is the monochromatic number of $\Gamma$. 
		Finally, from the construction of $\Gamma'$, the claim on the runtime follows immediately. 
\end{proof}

\subsection{\SETH-based Time Lower Bound}\label{sec:SETH}

In this subsection, we prove the following result.

\begin{theorem}\label{thm:SETH}
Let $F\colon \NN\to\NN$ be some computable increasing function.   Assuming randomized \SETH, for every $\varepsilon > 0$ and integer $k > 1$, no randomized $O(n^{k(1-\varepsilon )})$-time algorithm can decide  an instance $\Gamma(\C,c,c/F(k))$ of \intersection over universe $[n^{1+o(1)}]$, where $|\C|=n$.
\end{theorem}

Our proof builds on the following \SETH based lower bound for gap $k$-\maxcover proved in \cite{KLM19}.

\begin{theorem}[\cite{KLM19}]\label{thm:SETHmax}
	Let $F\colon \NN\to\NN$ be some computable increasing function.   Assuming \SETH, for every $\varepsilon > 0$ and integer $k > 1$, no randomized $O(n^{k(1-\varepsilon )})$-time algorithm can decide  an instance $\Gamma(G=(V\dot\cup W,E),1,1/F(k))$ of {\sf Unique} $k$-\maxcover. This holds even in the following setting:
\begin{itemize}
    \item $V:=V_1\dot\cup\cdots\dot\cup V_k$, where $\forall j\in[k]$, $|V_j|=n$.
    \item $W:=W_1\dot\cup\cdots\dot\cup W_\ell$, where $\ell=(\log n)^{O_{k}(1)}$ and $\forall i\in[k]$, $|W_i|=O_{k,\varepsilon}(1)$.
\end{itemize}
\end{theorem}
\begin{proof}[Proof Sketch]
The proof of the  theorem statement is by contradiction. Suppose there is a randomized $O(n^{k(1-\varepsilon )})$-time algorithm that can decide  every instance $\Gamma(G=(V\dot\cup W,E),1,1/F(k))$ of $k$-\maxcover for some fixed constant $\varepsilon > 0$ and integer $k > 1$.
All the references here are using the labels in~\cite{KLM19}. First we apply Proposition 5.1 to Theorem 6.1 with $z=\log_2 (F(k))$ to obtain an $(m/\alpha,O_k(\log_2 m),O_{k,\varepsilon}(1),1/F(k))$-efficient protocol for $k$-player $\mathsf{Disj}_{m,k}$ in the SMP model. 
The proof of the theorem then follows by plugging in the parameters of the protocol  to Corollary 5.3. To note that the instance constructed is that of {\sf Unique} $k$-\maxcover, see the remarks in Appendix B. 
\end{proof}

We now return to the proof of Theorem~\ref{thm:SETH}.

\begin{proof}[Proof of Theorem~\ref{thm:SETH}]
Fix $F\colon \NN\to\NN$.
Suppose there is a randomized $O(n^{k(1-\varepsilon )})$-time algorithm that can decide  every instance $\Gamma(\C,c,c/F(k))$ of \intersection  over universe $[n^{1+o(1)}]$ (where $|\C|=n$)  for some fixed constant $\varepsilon > 0$ and integer\footnote{The case $k=2$ can be easily handled here by standard input subdividing tricks used previously in~\cite{R18,KM20}. At the same time the case $k=2$ was already proved in~\cite{KM20}. } $k > 2$. We claim that  the algorithm can be used to solve every hard instance $\Gamma'(G=(V\dot\cup W,E),1,1/F(k))$ of $k$-\maxcover, as given in Theorem~\ref{thm:SETHmax}, in time $O(n^{k(1-\varepsilon)})$ where 
\begin{itemize}
    \item $V:=V_1\dot\cup\cdots\dot\cup V_k$, where $\forall j\in[k]$, $|V_j|=n$.
    \item $W:=W_1\dot\cup\cdots\dot\cup W_\ell$, where $\ell=(\log n)^{O_{k}(1)}$ and $\forall i\in[k]$, $|W_i|=O_{k,\varepsilon}(1)$.
\end{itemize}
This would then contradict Theorem~\ref{thm:SETHmax}.

Fix $\Gamma'(G=(V\dot\cup W,E),1,1/F(k))$. By applying Proposition~\ref{prop:mctointer} to $\Gamma'$ we obtain an instance $\Gamma''(\C_1,\ldots ,\C_k,\ell,\ell/F(k))$  of panchromatic \intersection over universe of size $O_\varepsilon((\log n)^{O_k(1)})$ with monochromatic number also bounded above by $c_{k,\varepsilon} \cdot \ell$ for some constant $c_{k,\varepsilon}$ depending only on $k$ and $\varepsilon$. 

Let $m:=\sqrt{n}$.
In Theorem~\ref{thm:pan_cons}, let $i^*\in\NN$ be such that $m\le n_{i^*}\le 2^{k}\cdot m$.  We sample $w:=\widetilde\Omega(4(D^k)!)$ many graphs $G_1,\ldots ,G_w$ from $\mathcal{D}_{k,c_{k}\cdot F(k),n_{i^*}}$ in time $O_k(n)$. By Theorem~\ref{thm:pan_cons}, we know that one of these graphs is a $(n_{i^*},n_{i^*},k,D^k,D^k/(c_{k}\cdot F(k)),(4(D^k)!)^{-1})$\nobreakdash-panchromatic graph with high probability and we find it in time $w\cdot n_{i^*}^{k+1}=O_k(n^{\frac{k}{2}+1})$. Let $G^*$ be one of the sampled graphs which is a $(n_{i^*},n_{i^*},k,D^k,D^k/(c_{k}\cdot F(k)),(4(D^k)!)^{-1})$-panchromatic graph. We randomly delete $n_{i^*}-m$ many vertices  in each colour class of $G^*$ to  obtain a $(m,n_{i^*},k,D^k,D^k/(c_{k}\cdot F(k)),(4(D^k)!)^{-1})$\nobreakdash-panchromatic graph. 

For every $i\in[k]$, arbitrarily equipartition $\mathcal{C}_i$ into $\mathcal{C}_i^1,\ldots ,\mathcal{C}_i^{m}$. 
 Given  $\Gamma''(\C_1,\ldots ,\C_k,\ell,\ell/F(k))$ we show how to construct $n^{k/2}$ instances $$\{\Gamma_{(t_1,\ldots ,t_k)}(\C,c,c/F(k))\}_{(t_1,\ldots, t_k)\in [m]^k},$$ of \intersection over universe $[n^{\frac{1}{2}+o(1)}]$ (where $|\C|=mk$). For every ${(t_1,\ldots ,t_k)}\in [m]^k$, define an instance  $\Gamma''_{(t_1,\ldots ,t_k)}(\C_1^{t_1},\ldots ,\C_k^{t_k},\ell,\ell/F(k))$ of panchromatic \intersection over universe of size $O_\varepsilon((\log n)^{O_k(1)})$ with monochromatic number also bounded above by $c_{k,\varepsilon} \cdot \ell$.

Fix ${(t_1,\ldots ,t_k)}\in [m]^k$. We apply Theorem~\ref{thm:PGC} to $\Gamma''_{{(t_1,\ldots ,t_k)}}$ by using $G^*$. We thus obtain an instance $\Gamma_{(t_1,\ldots ,t_k)}(\C,c:=\ell\cdot D^k,\max((\ell/F(k))\cdot D^k,\ell\cdot D^k/F(k))$ of \intersection over universe $\U$ in time $\tilde{O}(n^{1+o(1)})$ where $|\U|=m\cdot (\log n)^{O_k(1)}$. Also note that $|\C|=mk$.  

Thus, if $\Gamma'$ was in the completeness case then there exists ${(t_1,\ldots ,t_k)}\in [m]^k$ such that $\Gamma''_{(t_1,\ldots ,t_k)}$ is also in the completeness case, and consequently, $\Gamma_{(t_1,\ldots ,t_k)}$ is  in the completeness case. On the other hand, if $\Gamma'$ was in the soundness case then for every ${(t_1,\ldots ,t_k)}\in [m]^k$ we have that $\Gamma''_{(t_1,\ldots ,t_k)}$ is also in the soundness case, and consequently, $\Gamma_{(t_1,\ldots ,t_k)}$ is  in the soundness case too. The total runtime of the algorithm would be $n^{k/2}\cdot \left(n^{k(1-\varepsilon)/2}+n^{1+o(1)}\right)+n^{\frac{k}{2}+1}=O(n^{k(1-\frac{\varepsilon}{2})})$.
\end{proof}

\subsection{\ETH-based Time Lower Bound}\label{sec:ETH} 

In this subsection, we prove the following result.

\begin{theorem}\label{thm:ETH}
Let $F\colon \NN\to\NN$ be some computable increasing function.   Assuming randomized \ETH, for  sufficiently large integer $k$, no randomized $n^{o(k)}$-time algorithm can decide  an instance $\Gamma(\C,c,c/F(k))$ of \intersection over universe $[n^{1+o(1)}]$, where $|\C|=n$.
\end{theorem}

Our proof builds on the following \ETH based lower bound for gap $k$-\maxcover proved in \cite{KLM19}.

\begin{theorem}[\cite{KLM19}]\label{thm:ETHmax}
Let $F\colon \NN\to\NN$ be some computable increasing function.   Assuming \ETH, for sufficiently large integer $k$, no randomized $n^{o(k)}$-time algorithm can decide  an instance $\Gamma(G=(V\dot\cup W,E),1,1/F(k))$ of {\sf Unique} $k$-\maxcover. This holds even in the following setting:
\begin{itemize}
    \item $V:=V_1\dot\cup\cdots\dot\cup V_k$, where $\forall j\in[k]$, $|V_j|=n$.
    \item $W:=W_1\dot\cup\cdots\dot\cup W_\ell$, where $\ell=(\log n)^{O_{k}(1)}$ and $\forall i\in[k]$, $|W_i|=O_k(1)$.
\end{itemize}
\end{theorem}
\begin{proof}[Proof Sketch]
Suppose there is a randomized $n^{o(k)}$-time algorithm that can decide  every instance $\Gamma(G=(V\dot\cup W,E),1,1/F(k))$ of $k$-\maxcover for every $k\in\NN$.
All the references here are using the labels in~\cite{KLM19}. First we apply Proposition 5.1 to Theorem 7.1 with $z=\left(\log_2\frac{-1}{1-\delta}\right)\log_{2} (F(k))$ to obtain a $(0,O_k(\log_2 m),O_{k}(t),1/F(k))$-efficient protocol for $k$-player $\mathsf{MultEq}_{m,k,t}$ in the SMP model. 
The proof of the theorem then follows by plugging in the parameters of the protocol  to Corollary 5.4. To note that the instance constructed is that of {\sf Unique} $k$-\maxcover, see the remarks in Appendix B. 
\end{proof}

We now return to the proof of 
Theorem~\ref{thm:ETH}.

\begin{proof}[Proof of Theorem~\ref{thm:ETH}]
Fix $F\colon \NN\to\NN$.
Suppose there is a randomized $n^{o(k)}$-time algorithm that can decide  every instance $\Gamma(\C,c,c/F(k))$ of \intersection over universe $[n^{1+o(1)}]$ (where $|\C|=n$)  for every $k\in\NN$. 
Notice that such an algorithm can also be used to device a search that finds a witness in the YES case by making $nk$ calls to the decision algorithm. 

We claim that then this search algorithm can be used to solve (with high probability) every instance $\Gamma'(G=(V\dot\cup W,E),1,1/F(k))$ of $k$-\maxcover in time $O(n^{o(k)})$ where 
\begin{itemize}
    \item $V:=V_1\dot\cup\cdots\dot\cup V_k$, where $\forall j\in[k]$, $|V_j|=n$.
    \item $W:=W_1\dot\cup\cdots\dot\cup W_\ell$, where $\ell=(\log n)^{O_{k}(1)}$ and $\forall i\in[k]$, $|W_i|=O_{k}(1)$.
\end{itemize}
This would then contradict Theorem~\ref{thm:ETHmax}. 


Fix $\Gamma'(G=(V\dot\cup W,E),1,1/F(k))$. By applying Proposition~\ref{prop:mctointer} to $\Gamma'$ we obtain an instance $\Gamma''(\C_1,\ldots ,\C_k,\ell,\ell/F(k))$  of panchromatic \intersection over universe of size $(\log n)^{O_k(1)}$ with monochromatic number also bounded above by $c_{k} \cdot \ell$, for some constant $c_k$ only depending on $k$. 

In Theorem~\ref{thm:pan_cons}, let $i^*\in\NN$ such that $n\le n_{i^*}\le 2^k\cdot n$.  We sample $w:=\widetilde\Omega(4(D^k)!)$ many graphs $G_1,\ldots ,G_w$ from $\mathcal{D}_{k,c_{k}\cdot F(k),n_{i^*}}$ in time $O_k(n^2)$. By Theorem~\ref{thm:pan_cons}, we know that one of these graphs is a $(n_{i^*},n_{i^*},k,D^k,D^k/(c_{k}\cdot F(k)),(4(D^k)!)^{-1})$\nobreakdash-panchromatic graph with high probability. Next, in each of these $w$ many graphs, we randomly delete $n_{i^*}-n$ vertices  in each colour class. Note that in every $(n_{i^*},n_{i^*},k,D^k,D^k/(c_{k}\cdot F(k)),(4(D^k)!)^{-1})$\nobreakdash-panchromatic graph if we randomly delete $n_{i^*}-n$ vertices  in each colour class then we obtain a $(n,n_{i^*},k,D^k,D^k/(c_{k}\cdot F(k)),(4(D^k)!)^{-1})$\nobreakdash-panchromatic graph. 

Let $i\in[w]$. For each $G_i$ we apply Theorem~\ref{thm:PGC} to $\Gamma''$ by using $G_i$. If $G_i$ is a  $(n,n_{i^*},k,D^k,D^k/(c_{k}\cdot F(k)),(4(D^k)!)^{-1})$-panchromatic graph then we obtain an instance $\Gamma(\C,c:=\ell\cdot D^k,\max((\ell/F(k))\cdot D^k,\ell\cdot D^k/F(k))$ of \intersection over universe $\U$ in time $O(n^{2+o(1)})$ where $|\U|=n\cdot (\log n)^{O_k(1)}$. Also note that $|\C|=nk$.  

On the other hand, if $G_i$ was not a  $(n,n_{i^*},k,D^k,D^k/(c_{k}\cdot F(k)),(4(D^k)!)^{-1})$-panchromatic graph then we still obtain some instance of \intersection and the search algorithm would then output a witness if we are in the YES case of \intersection, which would not yield any meaningful solution to $\Gamma'$, and so we can discard it. 
\end{proof}

\section{Open Problems}\label{sec:open}

In this section, we highlight a few open problems. 

\subsection*{Closest Pair} In~\cite{KM20}, the authors constructed two kinds of panchromatic graphs\footnote{See \cref{km20}.}. First they constructed  $(n,m:=\polylog(n),2,t:=m^{\Omega(1)},t/\log n,1/n^{o(1)})$-panchromatic graphs by using the density and distance properties of low degree univariate polynomials. They also constructed $(n,\Theta(\log n),2,t:=\Omega(\log n),t(1-\varepsilon),1/\sqrt{n})$-panchromatic graphs (for some small constant $\varepsilon>0$) by using the density and distance properties of algebraic-geometric codes. The latter was used to prove conditional hardness of approximation results for 
the closest pair problem, where we are a set of $n$ points in $\mathbb{R}^d$ and we would like the closest pair of points in the  $\ell_p$-metric. Using the latter panchromatic graph, the authors showed that assuming \SETH, no  algorithm running in $n^{1.5-\delta(\varepsilon)}$ time can approximate the  closest pair problem to $(1+\varepsilon)$-factor. If there existed a $(n,m:=n^{o(1)},2,t:=\Omega(m),t(1-\varepsilon),1/n^{o(1)})$-panchromatic graph then it could prove the  subquadratic time inapproximability result for the closest pair problem\footnote{Both the panchromatic graphs constructed in~\cite{KM20} have the additional important property that they are biregular which is needed for proving lower bounds for the closest pair problem.}.

\subsection*{Hardness of \mincover.}
In Theorem~\ref{thm:ETH} we obtain tight running time lower bound for \intersection under \ETH but our inapproximability factor is weaker than the one ruled out by Lin~\cite{Lin18}. In Appendix~\ref{sec:mincover} we show a gap creating reduction for \intersection which starts from an instance of \mincover and reduces it to gap \intersection matching the inapproximability factors of~\cite{Lin18}.
Also, a tight running time lower bound is known for exact panchromatic \mincover under \ETH~\cite{KN21}. Is it possible to tweak our \PGC technique and use our construction of panchromatic graphs or design new panchromatic graphs or both, in order to reduce panchromatic \mincover to \mincover? If yes, then we could obtain a tight running time lower bound for \intersection under \ETH with inapproximability factors matching~\cite{Lin18}.

\subsection*{Biclique} Using a more intricate composition technique  and weaker objects than our threshold graphs, Lin~\cite{Lin18} showed that $k$-{\sf Biclique} problem is \W[1]-hard; in the $k$-{\sf Biclique} problem, we are given as input a balanced bipartite graph on $n$ vertices and the goal is to determine if it contains a $K_{k,k}$. Lin further showed that under \ETH, no $n^{o(\sqrt{k})}$ time algorithm can decide $k$-{\sf Biclique}.
However, if $(n,n,k,t:=O(k)),t-1,1/n)$-threshold graphs exist then we could obtain the tight time lower bound for  $k$-{\sf Biclique} under \ETH. Do such threshold graphs exist?

\subsection*{Derandomization} In this paper, we provide distributions from which we can efficiently sample panchromatic and threshold graphs. A natural derandomization question is to ask for explicit panchromatic and threshold graphs. 

\subsection*{Other Applications of Our Threshold Graphs} 
Norm-graphs have various applications in theoretical computer science such as proving lower bounds for span-programs~\cite{BGKRSW96,G01}, rectifier networks~\cite{Jukna13},  circuit lower bounds~\cite{JuknaS13}, and so on. But in each of these cases our threshold graph match the lower bound obtained by using norm-graphs. Is there an application in TCS where the stronger completeness property of threshold graphs comes in handy?
Also, somewhat speculatively, can our construction of (adjacency) matrices yield (semi-explicit) rigid matrices? If yes, this would be an excellent followup to~\cite{GT18}.

\subsection*{Other Applications of Our Panchromatic  Graphs} 
Our Panchromatic Graph Composition technique might be relevant with appropriate modifications to resolve various important complexity theoretic questions, such as the dichotomy conjecture of 
\cite{Grohe07} whose coloured variant was shown in \cite{CGL17}. 

\section*{Acknowledgements}

Boris Bukh was supported in part by U.S.\ taxpayers through NSF CAREER grant DMS-1555149. Karthik C.\ S.\ was financially supported by Subhash Khot’s Simons Investigator Award and by a grant from the Simons Foundation, Grant Number 825876, Awardee Thu D. Nguyen. Bhargav Narayanan was supported by NSF grants CCF-1814409 and DMS-1800521.

\bibliographystyle{alpha}
\bibliography{rand_alg_hardness}

\appendix

\section{From exact k\textsf{-MinCoverage} to gap k\textsf{-SetIntersection} via \textsf{TGC} technique}
\label{sec:mincover}

In this section, we generalize a  gap creation technique first appearing in \cite{Lin18}. 

\begin{theorem}[Generalization of Lin's Gap Creation technique from~\cite{Lin18}]
  There is an algorithm that given as input
  \begin{enumerate}
      \item an instance $\Gamma(\C,c,s)$ of \mincover over universe $[n]$, and
      \item an (n,m,c,t,r,1)-threshold graph $H(A,B)$, with $|A|=n$ and $|B|\le m$,
  \end{enumerate}
  then outputs an instance $\Gamma'(\C',t,r)$ of \intersection over universe $\U$ such that the following hold:
  \begin{description}
  \item[Size] $|\C'|=|\C|$ and $|\U|=|B|$.
\item[Completeness] If there exists $k$ sets $S_{i_1},\ldots ,S_{i_k}$ in $\C$ such that \[\left|\underset{{r\in[k]}}{\bigcup}S_{i_r}\right|\le c,\] then there exists $k$ sets $S_{i_1}',\ldots ,S_{i_k}'$ in $\C'$ such that \[\left|\underset{{r\in[k]}}{\bigcap}S_{i_r}'\right|\ge t,\]
	\item[Soundness] If for every $k$ sets $S_{i_1},\ldots ,S_{i_k}$ in $\C$ we have \[\left|\underset{{r\in[k]}}{\bigcup}S_{i_r}\right|\ge s,\] then for every $k$ sets $S_{i_1}',\ldots ,S_{i_k}'$ in $\C'$ we have \[\left|\underset{{r\in[k]}}{\bigcap}S_{i_r}'\right|\le r,\]
	\item[Running Time] The reduction runs in $\tilde{O}(n^2 m)$ time.
  \end{description}
\end{theorem}
\begin{proof}
We need to first define the edge set $E$ of the output bipartite graph $G$. Let $\sigma\colon \C'\to \C$ and $\pi:[n]\to A$ be some canonical one-to-one mappings. We include  in $S'\in \C'$ the universe element $u\in \U=B$ if and only if for every element $i_j$ in $\sigma(S'):=\{i_{1},\ldots ,i_{d}\}\subset [n]$, there is  an edge between $\pi(i_j)$ and $u\in B$ in the graph graph $H$. 

We analyze the completeness case by assuming there exists $k$ sets $S_{i_1},\ldots ,S_{i_k}$ in $\C$ such that \[\left|\underset{{r\in[k]}}{\bigcup}S_{i_r}\right|\le c.\] We claim that the $k$ sets $\sigma^{-1}(S_{i_1}),\ldots, \sigma^{-1}(S_{i_k})$ in $\C'$ have at least intersection size of $t$. Let $S:=\underset{{r\in[k]}}{\cup}S_{i_r}$ (where $|S|\le c$). Let $\hat S:=\{\pi(i)\mid i\in S\}\subset A$. Let $T\subset B$ be the set of common neighbors of $\hat S$ in $H$. 

From the threshold graph property of $H$, we have that $|T|\ge t$. We claim that every element in $T$ is contained in every set in $\{\sigma^{-1}(S_{i_1}),\ldots, \sigma^{-1}(S_{i_k})\}$. To see this, fix $u\in T$ and $j\in [k]$. Since $u$ is a common neighbor of $\hat S$ in $H$, it is also a common neighbor of every subset of $\hat S$ in $H$. Thus, $u$ is contained in $  \{\pi(i)\mid i\in S_j\}$.

Next consider the soundness case by assuming that for every $k$ sets $S_{i_1},\ldots ,S_{i_k}$ in $\C$ we have \[\left|\underset{{r\in[k]}}{\bigcup}S_{i_r}\right|\ge s.\] Consider any $k$ sets $S_1',\ldots ,S_k'$ in $V$ and fix an arbitrary universe element $u\in \U$. 

We have that $u$ is contained in the all the sets in $\{S_1',\ldots ,S_k'\}$ if and only if $u$ is a common neighbor of $\sigma(S_j')$ (and then applying $\pi$ on each of elements of $\sigma(S_j')$) in $H$ for every $j\in[k]$. In other words, $u$ is a common neighbor of $\underset{j\in[k]}{\cup} \pi\circ\sigma(S_j')$ in $H$. But we know from the soundness case assumption that \[\left|\underset{j\in[k]}{\bigcup} \pi\circ\sigma(S_j')\right|\ge s\ge c+1.\] From the threshold graph soundness property of $H$ we then have that $\underset{j\in[k]}{\cup} \pi\circ\sigma(S_j')$ can have at most $r$ common neighbors in $H$. Thus, $\{S_1',\ldots ,S_k'\}$ have at most intersection size of $r$.
\end{proof}

Finally, we note that an instance $\Gamma(\C,k,k+1)$ of \mincover over universe $[n]$ is \W[1]-hard to decide (follows from a straightforward reduction from the \clique problem).

\end{document}